\documentclass[12pt,journal,onecolumn]{IEEEtran}
\usepackage{amsmath, bm}
\usepackage{amsfonts}
\usepackage{amssymb}
\usepackage{theorem}
\usepackage{graphicx} 
\usepackage{xcolor}
\usepackage{cleveref}

\date{}
\usepackage{xcolor}
\newcommand{\comment}[1]{}
\newtheorem{theorem}{Theorem}
\newtheorem{lemma}[theorem]{Lemma}

\newtheorem{corollary}[theorem]{Corollary}
\newtheorem{definition}{Definition}
\newtheorem{remark}{Remark}

\begin{document}
	 	\IEEEoverridecommandlockouts
	\title{Second-Order Asymptotics of Hoeffding-Like Hypothesis Tests}
	\author{
		\IEEEauthorblockN{K. V. Harsha, Jithin~Ravi,\IEEEauthorrefmark{2} and Tobias~Koch\IEEEauthorrefmark{3}}\\
		\IEEEauthorblockA{\IEEEauthorrefmark{2}%
			Indian Institute of Technology Kharagpur, India\\
			\IEEEauthorrefmark{3}%
			Universidad Carlos III de Madrid, Legan\'es, Spain and 	Gregorio Mara\~n\'on Health Research Institute, Madrid, Spain.\\
			Emails: {kv.harsha1989@gmail.com, jithin@ece.iitkgp.ac.in, koch@tsc.uc3m.es}}
		\thanks{Part of this work was done while K.~V.~Harsha and J.~Ravi were with the Signal Theory and Communications Department,
			Universidad Carlos III de Madrid, Legan\'es, Spain. K.~V.~Harsha, J.~Ravi, and T.~Koch have received funding from the European Research Council (ERC) under the European Union's Horizon 2020 research and innovation programme (Grant No.~714161). T.~Koch has further received funding from the Spanish Ministerio de Ciencia e Innovaci\'on under Grant PID2020-116683GB-C21~/~AEI~/~10.13039/501100011033.}
	}
\maketitle
		\begin{abstract}
		We consider a  binary statistical hypothesis testing problem, where $n$ independent and identically distributed random variables $Z^n$ are either distributed according to the null hypothesis $P$ or the alternate hypothesis $Q$, and only $P$ is known. For this problem, a well-known test is the Hoeffding test, which accepts $P$
	if the Kullback-Leibler (KL) divergence between the empirical distribution of $Z^n$  and $P$ is below some threshold. In this paper, we consider Hoeffding-like tests, where the KL divergence is replaced by other divergences, and  characterize, for a large class of divergences, the first and second-order terms of the type-II error for a fixed type-I error. Since the considered class includes the KL divergence, we obtain the second-order term of the Hoeffiding test as a special case.
	\end{abstract}
\section{Introduction}
\label{sec:intro}
Statistical hypothesis testing is known to have applications in areas such as  information theory, signal processing, and machine learning. The most simple form of hypothesis testing is binary hypothesis testing, where the goal is to determine the distribution of a random variable $Z$ between a null hypothesis $P$ and an alternate hypothesis $Q$.  There can be two types of errors in a binary hypothesis testing: The type-I error is the probability of declaring the hypothesis as $Q$ when the true distribution  is $P$. The type-II error is the probability of declaring the hypothesis as $P$ when the true distribution is $Q$. In general, we are interested in analyzing the trade-off between these two types of errors for a given test. This test may have full or only partial access to the distributions $P$ and $Q$.

When both $P$ and $Q$ are known, the likelihood ratio test (also known as Neyman-Pearson test~\cite{NP33}) achieves the optimal error trade-off between type-I and type-II errors. The Neyman-Pearson test is also investigated in an asymptotic setting, where one observes  $n$ independent copies of $Z$ and the errors are analyzed asymptotically as $n$ tends to infinity.  In this case, we are often interested in the behaviour of the type-II error $\beta_n$ for a fixed type-I error $\alpha_n$. It is known that for $\alpha_n \leq \epsilon, \; \epsilon \in (0,1)$,  the corresponding type-II error satisfies \cite[Prop. 2.3]{T214} 
\begin{eqnarray}
	-\ln \beta_{n}  &=& nD(P \| Q) -  \sqrt{nV(P \| Q) }  Q^{-1}(\epsilon) + o(\sqrt{n}), \quad \text{as} \quad n \rightarrow \infty \label{eq:NP}
\end{eqnarray}
where  $D(P \| Q)$  is the Kullback-Leibler (KL) divergence between $P$ and $Q$, defined as
\begin{eqnarray}
	D(P \| Q)  \triangleq \sum_{i=1}^{k} P_{i} \ln \frac{P_{i}}{Q_{i}}, \label{eq:kl}
\end{eqnarray}
$Q^{-1}(\cdot)$ denotes the inverse  of the tail probability  of the standard Normal distribution, 
\begin{eqnarray*}
	V(P \| Q) \triangleq \sum_{i=1}^{k} P_{i} \left[ \left( \ln  \frac{P_{i}}{Q_{i}} -D(P \|Q) \right)^{2}  \right], 
\end{eqnarray*}
and $o(\sqrt{n})$ denotes a term that grows more slowly than $\sqrt{n}$ as $n \to \infty$. Thus the first-order term of $-\ln \beta_n$, defined as $\lim_{n\rightarrow \infty} \frac{-\ln \beta_n}{n}$ and sometimes referred to as the \textit{error exponent}, is given by $D(P\|Q)$.  The second-order term, defined as $\frac{-\ln \beta_n -n D(P\|Q)} {\sqrt{n}}$, is given by $-\sqrt{V(P \| Q) }  Q^{-1}(\epsilon)$.

The case where the test has only partial access to the distributions $P$ and $Q$ is generally studied under the name of {\em composite hypothesis testing}. Here, the assumption is that $P$ and $Q$ belong to certain classes of distributions denoted by $\mathcal{P}$ and $\mathcal{Q}$, respectively. Several special cases of composite hypothesis testing have been investigated in the literature. Two such cases are:  1) $P$ is known and $Q$ can be any distributions but $P$; 2)  $P$ is known and $Q$ belongs to a special class $\mathcal{Q}$. A test proposed by Hoeffding \cite{H65}, known as the Hoeffding test, is suitable for these two cases. In the Hoeffding test, the null hypothesis $P$ is accepted if the KL divergence between the type $t_{Z^{n}}$ (the empirical distribution) of the observations $Z^{n}=(Z_{1},\cdots,Z_{n})$ and $P$ is below some threshold $c$. Otherwise, the alternative hypothesis is accepted. Mathematically, 
\begin{align*}
	\text{if}	\; D(t_{Z^{n}} \|P) < c, \text{ then accept } H_0; \text{ otherwise accept } H_1. 
\end{align*}
In~\cite[Th.~III.2]{C98}, it  was shown that  the Hoeffding test achieves the same first-order term of $-\ln \beta_n$ as the Neyman-Pearson test. Consequently, not having access to the distribution of the alternate hypothesis does not affect the first-order term. To analyze the second-order term of $-\ln \beta_n$ in the above two cases of composite hypothesis testing, Watanabe proposed a test that is second-order optimal in some sense~\cite{Watanabe18}.

The  Hoeffding test does not require side information about
the special class $\mathcal{Q}$ to which $Q$ belongs. It is therefore perhaps not surprising that tests that have  knowledge of $Q$ (case 2 of composite hypothesis testing) may outperform the Hoeffding test. Such tests include the generalized likelihood ratio test (GLRT)~\cite{Zeitouni92} and test via mismatched divergence~\cite{Unnikrishnan11}. For example, it has been observed through simulations that the test via mismatched divergence performs better than the Hoeffding test \cite{Unnikrishnan11}. However, an analytical comparison between the Hoeffding test and other tests is missing since refined  asymptotics of the Hoeffding are not available.

In this paper, we wish to understand the behavior of the second-order term of the Hoeffding test as well as of Hoeffding-like tests where the KL divergence is replaced by other divergences,  in the following referred to  as \textit{divergence test}.  For a large class of divergences, we demonstrate that the second-order term is given by$-\sqrt{V(P \| Q) Q^{-1}_{\chi^{2},k-1}(\epsilon) }$, where $ Q^{-1}_{\chi^{2},k-1}(\cdot)$ denotes the inverse  of the tail probability of the chi-squared distribution with $k-1$ degrees of freedom.  This class of divergences includes the KL divergence, so we obtain the second-order term of the Hoeffding test, as a special case.

The rest of this paper is organized as follows. Section~\ref{sec:problem} presents notations and the problem formulation. Section~\ref{sec:results} discusses our main results. Section~\ref{Sec_Pf_Main} presents the proof of our results. Section~\ref{sec:conclusion} concludes the paper with a summary and discussion of our results. 

\section{Notations and  Problem Formulation}\label{sec:problem}
\subsection{Notations}
Let $f(x)$ and $g(x)$ be  two real-valued functions.
For $a \in \Bbb R \cup \{\infty\}$, we write  $f(x)=O(g(x))$ as $x \rightarrow a$ if $\limsup_{x \to a} \frac{|f(x)|}{|g(x)|} < \infty$.  Similarly, we  write  $f(x)=o(g(x))$ as $x \rightarrow a$ if $\lim_{x \to a} \frac{|f(x)|}{|g(x)|} = 0$. Finally, we write $f(x) = \Theta(g(x))$ as $x \to a$	 if $f(x)=O(g(x))$ and $\liminf_{x \to a} \frac{|f(x)|}{|g(x)|} > 0$.

Let $\{a_{n}\}$ be a sequence of positive real numbers. We say that a sequence of random variables is  $X_{n}=o_{p}(a_{n})$ if for every $\epsilon >0$, we have 
\begin{eqnarray*}
	\lim_{n \rightarrow \infty} P\left(   \left|  \frac{X_{n}}{a_{n}}\right| >\epsilon \right) =0.
\end{eqnarray*}
Similarly, we say that $X_{n}=O_{p}(a_{n})$ if for every $\epsilon >0$ there exist $M_{\epsilon}>0$ and $N_{\epsilon} \in \Bbb N$ such that 
\begin{eqnarray*}
	P\left(   \left|  \frac{X_{n}}{a_{n}}\right|  >M_{\epsilon} \right) \leq \epsilon,\quad  n \geq N_{\epsilon}.
\end{eqnarray*}
\subsection{ Divergence and Divergence Test}

We consider a divergence test on a random variable $Z$ that takes value in a discrete set  $\mathcal{Z}=\{ a_{1}, \cdots, a_{k}\}$, where $ k \geq 2$. We denote the probability distribution function of $Z$ by a $k$-length vector $P=(P_{1},\cdots,P_{k})^{T}$, where
\begin{eqnarray*}
	P_{i} \triangleq \text{Pr}\{Z=a_{i}\}, \quad i=1,\cdots,k.
\end{eqnarray*}
Let 
\begin{align*}
	\bar{\mathcal{P}}(\mathcal{Z}) \triangleq \left\lbrace P=(P_{1}, \cdots, P_{k})^{T} \; | \; P_{i} \geq 0, \; i=1,\cdots,k  \mbox{ and }  \sum_{i=1}^{k} P_{i}=1 \right\rbrace 
\end{align*}
be the set of all such probability distributions.  The set of all positive  probability distributions with $P _{i} >0, \; i=1,\cdots,k $ on $\mathcal{Z}$ is denoted by  $\mathcal{P}(\mathcal{Z})$.  Let $Z^{n}=(Z_{1}, \cdots, Z_{n})$  denote an $n$-length  i.i.d.  sequence   of  $Z$. 
By denoting the realization of $Z^n$ by $z^n$, we define  the type distribution $t_{z^{n}}=(t_{z^{n}}(a_{1}), \cdots,  t_{z^{n}}(a_{k}))^{T}$ of the sequence $z^{n}$ as
\begin{eqnarray*}
	(t_{z^{n}}(a_{1}), \cdots,  t_{z^{n}}(a_{k}))^{T} \triangleq \left( \frac{N_{a_{1}| z^{n}}}{n}, \cdots,   \frac{N_{a_{k}| z^{n}}}{n}   \right) ^{T}
\end{eqnarray*}
where $N_{a_{i} | z^{n}}, \; i=1, \cdots, k$  denotes the number of times $a_{i}$ occurs in the sequence $z^{n}$.  Since $\sum_{i=1}^{k}P_{i}=1$, a probability vector $ P  = (P_{1}, \cdots, P_{k})^{T} \in	\mathcal{P}(\mathcal{Z})$ can be represented by any $(k-1)$ components of $P$. With this representation, the set $ \mathcal{P}(\mathcal{Z})$ can be identified as a $(k-1)$-dimensional manifold; see \cite{AB216}.  In this paper, we choose the coordinate system of the probability distribution $ P  = (P_{1}, \cdots, P_{k})^{T}$ as the first $(k-1)$ components  denoted by $\mathbf{P}=(P_{1}, \cdots, P_{k-1})^{T} $.

Given any two probability distributions $P,Q \in \mathcal{P}(\mathcal{Z})$, one can define a  non-negative binary function $\tilde{D}(P \| Q) $ called a {\em divergence} which represents  a measure of discrepancy between them. A divergence  is not necessarily  symmetric in its arguments and also need not satisfy the triangle inequality;  see \cite{AB216, E85} for more details. Mathematically, a divergence is defined as follows.
\begin{definition}
	Consider two points $P$ and $Q$ in $ \mathcal{P}(\mathcal{Z})$ whose coordinates are $\mathbf{P}=(P_{1}, \cdots, P_{k-1})^{T} $ and $\mathbf{Q}=(Q_{1}, \cdots, Q_{k-1})^{T} $.  A {\em divergence}  $ \tilde{D}(P \| Q) $ between  $P$ and $Q$ is a smooth function of $\mathbf{P}$ and $\mathbf{Q}$   (we may write $\tilde{D}(P \| Q)=\tilde{D}(\mathbf{P} \| \mathbf{Q}) $) satisfying the following conditions: 
	\begin{enumerate}
		\item $ \tilde{D}(P \| Q) \geq 0$ for any $P, Q \in \mathcal{P}(\mathcal{Z})$. 
		\item $\tilde{D}(P \| Q) =0$ if, and only if, $P=Q$.
		\item   The Taylor expansion of $\tilde{D}$ satisfies
		\begin{eqnarray*}
			\tilde{D}(\mathbf{P}+ \varepsilon \| \mathbf{P}  ) =\frac{1}{2} \sum_{i,j =1}^{k-1} g_{ij}(\mathbf{P})  \varepsilon_{i} \varepsilon_{j} +O(\| \varepsilon \|_{2}^{3}), \quad \text{as } \| \varepsilon \|_{2} \rightarrow 0,
		\end{eqnarray*}
		for some  $(k-1) \times (k-1)$-dimensional positive-definite matrix $G=(g_{ij})$ that depends on $\mathbf{P}$ and $\varepsilon=(\varepsilon_{1}, \cdots,\varepsilon_{k-1})^{T}$, and where $\| \varepsilon \|_{2} \triangleq \sqrt{\sum_{i=1}^{k-1} (\varepsilon_{i})^{2}}$ is the Euclidean norm of $\varepsilon$.
	\end{enumerate}	
\end{definition}
For any  
$\tilde{D}(P \| Q)=\tilde{D}(\mathbf{P} \| \mathbf{Q}) $, the partial derivatives of  $\tilde{D}(P \| Q)$ with respect to the first variable $\mathbf{P}=(P_{1}, \cdots, P_{k-1})^{T} $ are  given by $\frac{\partial}{\partial P_{i}} \tilde{D}(P \| Q)=\frac{\partial}{\partial P_{i}}\tilde{D}(\mathbf{P} \| \mathbf{Q}), \; i=1,\cdots,k-1$. Similarly, the partial derivatives of  $\tilde{D}(P \| Q)$ with respect to the second variable $\mathbf{Q}=(Q_{1}, \cdots, Q_{k-1})^{T} $ are  given by $\frac{\partial}{\partial Q_{i}} \tilde{D}(P \| Q)=\frac{\partial}{\partial Q_{i}}\tilde{D}(\mathbf{P} \| \mathbf{Q}), \; i=1,\cdots,k-1$. Since, by definition, every divergence  $\tilde{D}$ attains a minimum at $P =Q$,  it holds that
\begin{eqnarray}
	\left. \frac{\partial}{\partial P_{i}} \tilde{D}(P \| Q) \right|_{P =Q}=0, \quad i=1,\cdots, k-1 \label{eq:divpart1}\\
	\left. \frac{\partial}{\partial Q_{i}} \tilde{D}(P \| Q) \right|_{P =Q}=0, \quad i=1,\cdots, k-1. \notag 
\end{eqnarray}
See~\cite[(3.1)--(3.4)]{E85} for more details. 

There are many classes of divergences that are widely used in various fields of science and engineering; see, e.g., \cite{CZPA209} for more details. A well-known example of a divergence is the $f$-divergence defined as follows:
Let $f:(0,\in\infty) \rightarrow \Bbb R$ be a convex function with $f(1)=0$. Then, the \textit{$f$-divergence} between $P$ and $Q$ is defined as 
\begin{eqnarray}
	D_{f}(P \| Q) \triangleq \sum Q_{i} f\left( \frac{P_{i}}{Q_{i}}\right) 
	\label{eq:fdiv}. 
\end{eqnarray}
When  $f(u)=u \log u$, $D_{f}$ is the Kullback-Leibler divergence.
For 
\begin{eqnarray}
	f(u)= \frac{4}{1-\alpha^{2}} (1-u^{(1+\alpha)/2}), \; \text{where}  \; \alpha \neq \pm 1 \label{eq:confdiv}
\end{eqnarray}
the $f$-divergence is the  \textit{$\alpha$-divergence} given by
\begin{eqnarray}
	D_{\alpha}(P \| Q) \triangleq \frac{4}{1-\alpha^{2}} \left[  1- \sum  P_{i}^{\frac{1-\alpha}{2}} Q_{i}^{\frac{1+\alpha}{2}}  \right], \quad  \alpha \neq \pm 1 .  \label{eq:alphadiv}
\end{eqnarray}

We next define a divergence test for case 1 of composite hypothesis testing (see Section~\ref{sec:intro}), That is, it is known that under hypothesis $H_0$ the distribution is $P$, while  under hypothesis $H_1$ the distribution is anything but $P$. A  Hoeffding-like test or divergence test $T_{n}^{\tilde{D}}(r)$ for testing $H_{0} : \bar{P}=P$ against the alternative $H_{1}: \bar{P}=Q$ is defined as follows: \begin{quote} Upon observing  $z^{n}$, if $\tilde{D}(t_{z^{n}} \| P) < r$ for some $r>0$, then the null hypothesis $P$ is accepted; else $P$ is rejected. \end{quote} For $r >0$, define the acceptance region for $P$ as
\begin{eqnarray*}
	\mathcal{A}^{\tilde{D}}_{n}(r) \triangleq	\left\lbrace  z^{n} \; \left| \right.  \; \tilde{D}(t_{z^{n}} \| P)  <  r \right\rbrace.  
\end{eqnarray*}
Then,  the  type-I and the type-II errors are given by 
\begin{align*}
	\alpha_{n}(T_{n}^{\tilde{D}}(r)) & \triangleq P^{n} \left(   \mathcal{A}^{\tilde{D}}_{n}(r) ^{c} \right) \\
	\beta_{n}(T_{n}^{\tilde{D}}(r)) &  \triangleq Q^{n} ( \mathcal{A}^{\tilde{D}}_{n}(r) ).
\end{align*}
Our goal is to analyze  the asymptotics of the type-II error when the type-I error satisfies   $\alpha_{n} \leq \epsilon, \; 0<\epsilon<1$. Following the discussion  in Section~\ref{sec:intro},  we define the first-order term $\beta'$ and the second-order term $\beta''$ for the divergence test as follows:
\begin{eqnarray*}
	\beta' &\triangleq & \lim\limits_{n \rightarrow \infty} -\frac{1}{n}  \log \beta_{n} \\
	\beta'' &\triangleq &  \lim\limits_{n \rightarrow \infty} \frac{-\ln \beta_n -n \beta'} {\sqrt{n}}
\end{eqnarray*}
if the limits exist. For the Hoeffding test, it is known that $\beta'=D(P\|Q)$. In this paper, we generalize this result to the divergence test for a class of divergences defined in Definition~\ref{def:div_class}.  We further obtain $\beta''$ for the divergence test.
  \section{Second-order Asymptotics of the Divergence Test}\label{sec:results}
 We shall consider the following class of divergences.
\begin{definition}
	\label{def:div_class}
	Let $\varXi$ denote the class of divergences satisfying the following conditions:
	\begin{enumerate}
			\item  For any $P, T \in \mathcal{P}(\mathcal{Z})$, the  second-order  Taylor approximation of $\tilde{D} (T \| P)$  around $T=P$ is given by
		\begin{eqnarray}
			\tilde{D}(T \| P)
			&= &\eta  d_{\chi^{2}} (T, P) + O( \| \mathbf{T}-\mathbf{P} \|_{2}^{3}),\quad \text{as} \quad \| \mathbf{T}-\mathbf{P} \|_{2} \rightarrow 0 \label{eq:tylpfd}
		\end{eqnarray}
for some constant	$\eta>0$ (that possibly depends on $P$). In \eqref{eq:tylpfd},	 $\mathbf{T}=(T_{1},\cdots,T_{k-1})^{T}$,  $\mathbf{P}=(P_{1},\cdots,P_{k-1})^{T}$,    and  $d_{\chi^{2}}$ is the $\chi^{2}$-divergence defined as
		\begin{eqnarray}
			d_{\chi^{2}} (T, P)  \triangleq \sum_{i=1}^{k}  \frac{(T_{i} -P_{i})^{2}}{P_{i}}. \label{eq:chisquare}
		\end{eqnarray}
	\item For any $P, T \in \mathcal{P}(\mathcal{Z})$, the divergence $\tilde{D}$ satisfies the Pinsker-type inequality 
		\begin{eqnarray}
			\tilde{C}\|T-P\|^{2}_{1} \leq \tilde{D}(T\|  P), \quad \text{for some } \tilde{C}>0 \label{eq:pinskerfd}
		\end{eqnarray}
	where $\|T-P\|_{1} \triangleq \sum_{i=1}^{k} |T_i - P_i| $ denotes the $L^{1}$ norm of
	$T-P $.
		\item   The CDF of the statistic $  \eta^{-1} n	\tilde{D} (t_{z^{n}} \| P)$ can be approximated as 
		\begin{eqnarray}
			P^{n}( 	\eta^{-1} n	\tilde{D} (t_{Z^{n}} \| P) <c)=  F_{\chi^{2},k-1}(c) + O(\delta_{n}),\quad c>0 \label{eq:ratediv2d}
		\end{eqnarray}
		for some positive  sequence  $\delta_{n}$ satisfying  $\lim_{n \rightarrow \infty} \delta_{n}=0$, where $F_{\chi^{2},k-1}$ is the CDF of the chi-squared random variable $\chi^{2}_{k-1}$ with $k-1$ degrees of freedom.
	\end{enumerate}
\end{definition}
The following theorem characterizes $\beta'$ and $\beta''$ of the divergence test for any divergence $\tilde{D} $ in $ \varXi$.
 \begin{theorem} \label{divergence}
	Let $\tilde{D} \in \varXi$ and  $0< \epsilon < 1$.  Consider the divergence test $T_{n}^{\tilde{D}}(r)$  for testing $H_{0} : P$ against the alternative $H_{1}: Q$, where $P, Q \in \mathcal{P}(\mathcal{Z}) $ and  $P \neq Q$. Then,  the type-II error 
	satisfies the following:\\
	Part 1: There exists a threshold value $r_{n}$ satisfying 
	\begin{align}
		\alpha_{n}(T_{n}^{\tilde{D}}(r_{n})) \leq \epsilon \label{eq:typeepsilon}
	\end{align}
	such that, as $n\rightarrow \infty$,
	\begin{align}\label{eq:achiev}
		-\ln \beta_{n}(T_{n}^{\tilde{D}}(r_{n}))  & \geq   nD(P \| Q) -  \sqrt{nV(P \| Q) Q^{-1}_{\chi^{2},k-1}(\epsilon) } +O( \max\{ \delta_{n} \sqrt{n}, \ln n \}).
	\end{align}
	Part 2: For all $r_{n} >0$ satisfying \eqref{eq:typeepsilon}, we have as $n\rightarrow \infty$
	\begin{align}\label{eq:converse}
		-\ln \beta_{n}(T_{n}^{\tilde{D}}(r_{n}))  & \leq  nD(P \| Q) -  \sqrt{nV(P \| Q) Q^{-1}_{\chi^{2},k-1}(\epsilon) } +O( \max\{ \delta_{n} \sqrt{n}, \ln n \})
	\end{align}
	where $Q^{-1}_{\chi^{2},k-1}(\cdot)$ is the inverse  of  $c \mapsto 1-F_{\chi^{2},k-1}(c)$.
\end{theorem}
\begin{proof}
	See Section~\ref{Sec_Pf_Main}.
\end{proof}
Since the sequence $\delta_n$ in~\eqref{eq:achiev} and~\eqref{eq:converse} tends to zero as $n \to \infty$, we have that $O( \max\{ \delta_{n} \sqrt{n}, \ln n \}) = o(\sqrt{n})$. Consequently, \eqref{eq:achiev} and~\eqref{eq:converse} imply that, for the optimal threshold value $r_{n}$,
\begin{align}\label{eq:second_order}
-\ln \beta_{n}(T_{n}^{\tilde{D}}(r_{n})) & =  nD(P \| Q) -  \sqrt{nV(P \| Q) Q^{-1}_{\chi^{2},k-1}(\epsilon) } + o(\sqrt{n}).
\end{align}
Thus, Theorem~\ref{divergence} characterizes the first and second-order terms of the divergence test for any $\tilde{D} \in \varXi$. It can be verified that $\sqrt{ Q^{-1}_{\chi^{2},k-1}(\epsilon) } > Q^{-1}(\epsilon) $ for all $0 <\epsilon<1$. Consequently, the second-order term $\beta''$ of the divergence test is strictly smaller than the second-order term of the Neyman-Pearson test. 

For certain divergences in $\varXi$,  we can obtain more precise asymptotics than the one  in~\eqref{eq:second_order} by tightening the $o(\sqrt{n})$ term. To this end, we shall first discuss  the three conditions stated in Definition~\ref{def:div_class}, which we shall do one by one. 

\emph{Condition 1}:
 When $\tilde{D}$ is the $f$-divergence $D_{f}$,  it follows from a Taylor-series expansion of $D_{f}(T \| P)$ around $T=P$ that $D_{f}$ satisfies  \eqref{eq:tylpfd} with $\eta=\frac{f''(1)}{2}$ \cite[Th. 4.1]{CS204}.  Since  both the $\alpha$-divergence and the KL divergence $D$ belong to the $f$-divergence class with $f''(1)=1$, we obtain that $\tilde{D}=\{ D_{\alpha}, D\}$ satisfy \eqref{eq:tylpfd} with $\eta=\frac{1}{2}$ \cite{AC210}. 
 Furthermore, for the Renyi divergence of order $\alpha>0$, defined as
 \begin{eqnarray*}
 	I_{\alpha}(P \| Q) =  \frac{1}{\alpha-1} \left[  \ln \left(  \sum _{i=1}^{k} P_{i}^{\alpha} Q_{i}^{1-\alpha} \right)   \right], \quad \alpha \neq 1
 \end{eqnarray*}
 we can show that $\eta=\frac{\alpha}{2}, \; i=1,\cdots, k-1$. Thus,  $I_{\alpha}$ satisfies  \eqref{eq:tylpfd} with $\eta=\frac{\alpha}{2}$. There are many more divergences that satisfy \eqref{eq:tylpfd}, such as  the Battacharya divergence, the Tsallis divergence, and the  Sharma-Mittal divergence; see \cite{CZPA209} for more details. In information geometry, such divergences are referred to as invariant divergences \cite[Eq. 2.29]{CZPA209}.
 
\emph{Condition 2}: It is well-known that the KL divergence $D$ satisfies Pinsker's inequality \cite[Lemma 11.6.1]{ATCB}. From \cite[Th. 3]{G10}, it further follows that Pinsker's inequality can be generalized to many $f$-divergences under some conditions on $f$. In particular, it follows from \cite[Cor. 6]{G10} that  the $\alpha$-divergence satisfies a Pinsker-type inequality when $3\leq \alpha\leq 3, \alpha \neq \pm 1$:
\begin{eqnarray*}
	\frac{1}{2}\|T-P\|^{2}_{1} \leq D _{\alpha}(T\|  P). \label{eq:pinskerf}
\end{eqnarray*}
It follows from the same corollary that the Renyi divergence  $I_{\alpha}$ also satisfies  \eqref{eq:pinskerfd}.   
Furthermore, for any $f$-divergence, \cite[Th. 1]{C67} demonstrates that, if $D_{f} (T, P) <\epsilon$ for a sufficiently small $\epsilon$, then $\frac{f''(1)}{2 \sqrt{2}} \|T-P\|^{2}_{1} <\epsilon$. This implies, that if $D_{f}(T,P)=O(a_{n})$ for some positive sequence $a_{n}$ such that $a_{n}\rightarrow 0$, then $\|T-P\|_{1}=O(\sqrt{a_{n}})$, which is sufficient to prove Theorem~\ref{divergence}. 
In fact, condition 2 can be replaced by the condition that, if $\tilde{D}(T,P)=O(a_{n})$ for some positive sequence $a_{n}$ such that $a_{n}\rightarrow 0$, then $\|T-P\|_{1}=O(\sqrt{a_{n}})$.

\emph{Condition 3}: Note that this condition is a Berry-Esseen-type  convergence of the statistic $  \eta^{-1} n	\tilde{D} (t_{Z^{n}} \| P)$ to the chi-squared distribution with $k-1$ degrees of freedom. For both the $\alpha$-divergence and the KL divergence, we have $\eta=\frac{1}{2}$. The following result implies that, for these divergences, \eqref{eq:ratediv2d} holds with  $\delta_{n}=\frac{1}{\sqrt{n}}$.  Indeed, it follows from~\cite[Th.  3.1]{TR84} and~\cite[Th. 3]{Y72} that the power divergence statistics 	
	\begin{eqnarray*}
		T_{\lambda} (\mathbf{Y})  &\overset{\Delta}{=} & \frac{2}{\lambda (\lambda+1)}\sum_{i=1}^{k} Y_{i} \left[  \left(  \frac{Y_{i}}{nP_{i}}\right) ^{\lambda}-1 \right], \quad  \lambda \in \Bbb R \setminus \{0,-1\} \\
			T_{0} (\mathbf{Y})  &\overset{\Delta}{=} &	\lim_{\lambda \rightarrow 0}T_{\lambda} (\mathbf{Y}) \\
			T_{-1} (\mathbf{Y})  &\overset{\Delta}{=} &	\lim_{\lambda \rightarrow -1}T_{\lambda} (\mathbf{Y}) 
	\end{eqnarray*}
	satisfies
	\begin{eqnarray}
		P^{n} (T_{\lambda}(\mathbf{Y}) <c)=  F_{\chi^{2},k-1}(c) + O(n^{-1/2}),\quad c>0, \; \lambda \in \Bbb R. \label{eq:ratealphadiv1}
	\end{eqnarray}
See  \cite{UZ09} for more details. Letting  $\lambda=\frac{-(1+\alpha)}{2}, \; \alpha \neq \pm 1$, we can write 	$T_{\lambda} (\mathbf{Y}) $ in terms of the $\alpha$-divergence as	
	\begin{eqnarray*}
		T_{\lambda} (\mathbf{Y})  =
		2n D_{\alpha}(t_{Z^{n}} \| P), \quad \alpha \neq \pm 1.
	\end{eqnarray*}
Similarly, for $\lambda=0$, $T_{0} (\mathbf{Y}) $ can be expressed in terms of the KL divergence as $ T_{0} (\mathbf{Y})  =
2n D(t_{z^{n}} \| P)$.
It then follows from \eqref{eq:ratealphadiv1} that  $\tilde{D}=\{ D_{\alpha}, D\}, \; \alpha\neq  \pm 1$  satisfies \eqref{eq:ratediv2d}.   

From the above discussion, we conclude that $\tilde{D} = \{ D, D_{\alpha}\}$ with $-3\leq \alpha\leq 3, \alpha \neq \pm 1$  satisfy all  three conditions in Definition~\ref{def:div_class}  with $\delta_{n}=\frac{1}{\sqrt{n}}$ and thus  belong to the class  $\varXi$.  We thus have the following corollary to Theorem~\ref{divergence}, which sharpens the third-order term in~\eqref{eq:second_order}. 
 \begin{corollary} \label{alpdivergence}
	Let $0< \epsilon < 1$. Consider the divergence test $T_{n}^{\tilde{D}}(r)$ for $\tilde{D} = \{ D, D_{\alpha}\}$, $-3\leq \alpha\leq 3, \alpha \neq \pm 1$,  for testing $H_{0} : P$ against the alternative $H_{1}: Q$, where $P, Q \in \mathcal{P}(\mathcal{Z}) $ and  $P \neq Q$. Then,  the type-II error  of $T_{n}^{\tilde{D}}(r_{n})$ minimized over all the threshold values $r_{n}$ satisfying
	\begin{eqnarray*}
		\alpha_{n}(T_{n}^{\tilde{D}}(r_{n})) \leq \epsilon
	\end{eqnarray*}
	can be characterized as
	\begin{align*}
		-\ln \beta_{n}(T_{n}^{\tilde{D}}(r_{n}))  =  nD(P \| Q) -  \sqrt{nV(P \| Q) Q^{-1}_{\chi^{2},k-1}(\epsilon) } 
		+O(\ln n),\quad \text{as} \quad n\rightarrow \infty. 
	\end{align*}
\end{corollary}
Note that, when $\tilde{D}$ is the KL divergence, the divergence test $T_{n}^{\tilde{D}}(r)$ is the Hoeffding test. Hence, the above corollary recovers the second-order asymptotics of the Hoeffding test as a special case. 

\section{Proof of Theorem~\ref{divergence}}
\label{Sec_Pf_Main}
\subsection{Proof of  Part 1 of  Theorem  \ref{divergence}}
From  \eqref{eq:ratediv2d}, it follows that  there exist $M_{0}>0$ and $N_{0} \in \Bbb N$ such that 
\begin{eqnarray}
	\left| P^{n} \left( \eta^{-1}n \tilde{D}(t_{Z^{n}} \| P) \geq c \right) - Q_{\chi^{2},k-1}(c) \right|  \leq M_{0}\delta_{n}, \quad n \geq N_{0}. \label{eq:ratekl1}
\end{eqnarray} 
For $0<\epsilon<1$, let 
\begin{eqnarray}
	r_{n}= \frac{\eta}{n}Q_{\chi^{2},k-1}^{-1} \left( \epsilon- M_{0} \delta_{n} \right).  \label{eq:rvalue}
\end{eqnarray}
Then, from \eqref{eq:ratekl1}, it follows that the type-I error is upper-bounded by $\epsilon$:
\begin{eqnarray*}
	\alpha_{n}(T_{n}^{\tilde{D}}(r_{n})) & =&  P^{n}(  \tilde{D}(t_{Z^{n}} \| P)  \geq  r_{n}) \\
	& =&  P^{n} \left(  \eta^{-1} n D(t_{Z^{n}} \| P) \geq  Q_{\chi^{2},k-1}^{-1} ( \epsilon- M_{0} \delta_{n} ) \right) \\
	& \leq & Q_{\chi^{2},k-1} \left(  Q_{\chi^{2},k-1}^{-1} \left( \epsilon- M_{0} \delta_{n}  \right) \right) +  M_{0} \delta_{n}  \\
	& =& \epsilon, \quad n\geq N_{0}.
\end{eqnarray*}
We next consider the type-II error for  $r_n$ in \eqref{eq:rvalue}. To this end, define
\begin{eqnarray*}
	A_{\tilde{D}}(r') &\triangleq & \left\lbrace T \in \mathcal{P}(\mathcal{Z}) \mid \tilde{D} (T \|  P) <r' \right\rbrace, \quad r'>0. \label{eq:klball}  
\end{eqnarray*}
Then,  
\begin{eqnarray}
	\beta_{n}(T_{n}^{\tilde{D}}(r_{n}))  &=&  Q^{n}(\tilde{D}(t_{Z^{n}} \| P)   < r_{n}) \nonumber  \\
	&= &  \sum_{z^{n}:\tilde{D}(t_{z^{n}} \| P)   < r_{n} } Q^{n}(z^{n}) \nonumber  \\
	&=&  \sum_{\tilde{P} \in \mathcal{P}_{n} \cap A_{\tilde{D}} (r_{n})}  Q^{n}(T(\tilde{P})) \nonumber  \\
	& \leq&   \sum_{\tilde{P} \in \mathcal{P}_{n} \cap A_{\tilde{D}} (r_{n})}  \exp \{-n D(\tilde{P} \| Q) \},\quad n\geq N_{0}  \label{eq:tyerr}
\end{eqnarray}
where the last step follows from  \cite[Th. 11.1.4]{ATCB}. In~\eqref{eq:tyerr}, $\mathcal{P}_{n}$ denotes the set of types with denominator $n$ and $T(\tilde{P})$ is  the type class  of $\tilde{P}$,  defined as
\begin{eqnarray}
	T(\tilde{P}) \triangleq \left\lbrace z^{n} \in \mathcal{Z}^{n} \; | \; t_{z^{n}} = \tilde{P} \right\rbrace. \label{eq:typeclass}
\end{eqnarray}

We next derive a lower bound on $D(\tilde{P} \| Q)$ for $\tilde{P} \in \mathcal{P}_{n} \cap A_{\tilde{D}} (r_{n})$. To this end,  we shall use  the following auxiliary results.
\begin{lemma} \label{ctaylor} For any two probability distributions  $Q=(Q_{1},\cdots,Q_{k})^{T}$ and  $T=(T_{1},\cdots, T_{k})^{T}$, the second-order Taylor approximation  of the function $T \mapsto D(T \|Q)$ around  the null hypothesis $T=P$ is given by
	\begin{eqnarray}
		D(T \|Q)=  D(P \| Q) + \sum_{i=1}^{k} (T_{i} -P_{i} ) \ln    \left( \frac{P_{i}}{ Q_{i}} \right)+  \frac{1}{2} d_{\chi^{2}} (T, P) + O(\|\mathbf{T}-\mathbf{P}\|_{2}^{3})  \label{eq:tylq}
	\end{eqnarray}
	as $\|\mathbf{T}-\mathbf{P}\|_{2} \rightarrow 0$.
\end{lemma}
\begin{IEEEproof} See Appendix~\ref{claimtaylor}.
\end{IEEEproof}
By the  Pinsker-type inequality, we obtain that, for every $T \in A_{\tilde{D}}(r)$,  we have $\|T-P\|_{1} \leq \tilde{k} \sqrt{r}$ for some constant $\tilde{k}$.  This implies that $\|\mathbf{T}-\mathbf{P}\|_{1} \leq \tilde{k}_{1} \sqrt{r}$ for some constant $\tilde{k}_{1}$. Since any two norms on a finite-dimensional Euclidean space are equivalent, there exists a constant $C_{0} >0$ such that the Euclidean distance  is bounded by
\begin{eqnarray}
	\|\mathbf{T}-\mathbf{P}\|_{2} &\leq &  C_{0} \sqrt{r}. \label{eq:euorder}
\end{eqnarray}
Thus, for any $T \in A_{\tilde{D}}(r_{n})$ with $r_{n}$ in \eqref{eq:rvalue}, it follows  that 
\begin{eqnarray}
	\|\mathbf{T}-\mathbf{P}\|_{2}  = O\left( \frac{1}{\sqrt{n}}\right)
\end{eqnarray}
since $ r_{n}=O\left( \frac{1}{n}\right)$. Then, for any $T \in A_{\tilde{D}}(r_{n})$ with $r_{n}$ in \eqref{eq:rvalue},  \eqref{eq:tylpfd} and \eqref{eq:tylq} and  can be written as 
\begin{eqnarray}
		\tilde{D}(T \| P)
	&= &\eta d_{\chi^{2}} (T, P) + O\left( \frac{1}{n^{3/2}} \right) \label{eq:typorder} \\  
	D (T \|  Q) &=&  D(P \| Q) + \sum_{i=1}^{k} (T_{i} -p_{i} ) \ln    \left( \frac{p_{i}}{ q_{i}} \right)+  \frac{1}{2} d_{\chi^{2}} (T, P) + O\left( \frac{1}{n^{3/2}} \right).  \label{eq:tyqorder}
\end{eqnarray}
The expansion \eqref{eq:typorder} implies that  there exist $M>0$ and $N_{1} \geq  N_{0}$ such that, for all $ n \geq  N_{1} $ and $T \in A_{\tilde{D}}(r_{n})$,
\begin{eqnarray}
	|\tilde{D}(T \| P) -  \eta d_{\chi^{2}} (T, P) | \leq  \frac{ M }{n^{3/2}}. \label{eq:tylp2}
\end{eqnarray}
Next define
\begin{eqnarray*}
	A_{\chi^{2}}(r') &\triangleq & \left\lbrace T  \in \mathcal{P}(\mathcal{Z}) \mid d_{\chi^{2}} (T,  P) <r' \right\rbrace,\quad r'>0\\
	\bar{A}_{\chi^{2}}(r') &\triangleq & \left\lbrace T\in \mathcal{P}(\mathcal{Z}) \mid d_{\chi^{2}} (T,  P) \leq r' \right\rbrace, \quad r'>0. 
\end{eqnarray*}
We have the following lemma.
\begin{lemma} \label{lemmaball}  Let $ r_{n}$ be given in  \eqref{eq:rvalue}. Then, we have 
	\begin{eqnarray}
		A_{\tilde{D}}(r_{n}) \subseteq
		\bar{A}_{\chi^{2}} \left( \frac{r_{n}}{\eta}+ \frac{M}{ \eta n^{3/2}} \right), \quad   n \geq N_{1}. \label{eq:klchi}
	\end{eqnarray}
\end{lemma}
\begin{IEEEproof} 
	It follows from \eqref{eq:tylp2} that, for  $T \in A_{\tilde{D}}(r_{n})$  and $ n \geq N_{1}$, 
	\begin{eqnarray*}
		\eta d_{\chi^{2}} (T, P) - \frac{ M }{n^{3/2}} 
		&\leq& \tilde{D}(T \| P)<r_{n}.
	\end{eqnarray*}	
	This implies that
	\begin{eqnarray*}
		d_{\chi^{2}} (T, P) < \frac{r_{n}}{\eta}+ \frac{M}{ \eta n^{3/2}}.
	\end{eqnarray*}
	Thus $T \in A_{\chi^{2}} \left( \frac{r_{n}}{\eta}+ \frac{M}{ \eta n^{3/2}} \right)	 \subseteq \bar{A}_{\chi^{2}} \left( \frac{r_{n}}{\eta}+ \frac{M}{ \eta n^{3/2}} \right)	$.
\end{IEEEproof}

We next  use \eqref{eq:tyqorder} to lower-bound the term $D(\tilde{P} \| Q)$ for  $\tilde{P} \in \mathcal{P}_{n} \cap A_{D} (r)$.   Indeed, it follows from \eqref{eq:tyqorder} that there exist $M_{2}>0$ and $\ddot{N}_{2} \in \Bbb N$ such that, for all $n \geq  \ddot{N}_{2} $ and  $T \in A_{\tilde{D}}(r_n)$,
\begin{eqnarray*}
	|D (T \|  Q) -  D(P \| Q) - \sum_{i=1}^{k} (T_{i} -P_{i} ) \ln    \left( \frac{P_{i}}{ Q_{i}} \right)-  \frac{1}{2} d_{\chi^{2}} (T, P) | \leq \frac{ M_{2} }{n^{3/2}}. 
\end{eqnarray*}
Consequently, for all $n \geq  \ddot{N}_{2} $ and  $T \in A_{\tilde{D}}(r_{n})$,
\begin{eqnarray}
	D (T\|  Q) & \geq &  D(P \| Q) + \sum_{i=1}^{k} (T_{i} -P_{i} ) \ln    \left( \frac{P_{i}}{ Q_{i}} \right)+ \frac{1}{2} d_{\chi^{2}} (T, P)- \frac{ M_{2} }{n^{3/2}} \nonumber  \\
	& \geq &  D(P \| Q) + \sum_{i=1}^{k} (T_{i} -P_{i} ) \ln    \left( \frac{P_{i}}{ Q_{i}} \right)- \frac{ M_{2} }{n^{3/2}}  
	\label{eq:tylq4}
\end{eqnarray}
where the last inequality follows since $ d_{\chi^{2}} (T, P)$ is non-negative.

To proceed further, let us consider the following function,  defined for a probability distribution $\Gamma=(\Gamma_{1}, \cdots, \Gamma_{k})^{T} \in  \mathcal{P}(\mathcal{Z})$
\begin{eqnarray}
	\ell(\Gamma) 
	&\triangleq&\sum_{i=1}^{k} (\Gamma_{i} -P_{i} ) \alpha_{i} \label{eq:hfun1}
\end{eqnarray}
where
\begin{eqnarray}
	\alpha_{i}\triangleq\ln    \left( \frac{P_{i}}{ Q_{i}} \right),\quad i=1,\cdots,k. \label{eq:alphai}
\end{eqnarray}
For $n \geq N_{1} $ and  $\tilde{P} \in  \mathcal{P}_{n} \cap A_{\tilde{D}} ( r_{n})$, we have that 
\begin{eqnarray}
	\ell(\tilde{P}) \geq \min_{\mathcal{P}_{n} \cap A_{\tilde{D}} ( r_{n})} \ell(\Gamma)  \geq \min_{\bar{A}_{\chi^{2}} \left(   \frac{r_{n}}{\eta}+ \frac{M}{ \eta n^{3/2}} \right) } \ell(\Gamma)  \label{eq:min}
\end{eqnarray}
where the last  inequality follows from \eqref{eq:klchi}.

To compute the right-most minimum in \eqref{eq:min}, we define the following quantities:
\begin{eqnarray} 
	\mathcal{I} &\triangleq& \left\lbrace  i=1, \cdots,k \; | \;  \alpha_{i} -D(P \| Q)>0  \right\rbrace \label{eq:tauin} \\
	\mathcal{B} &\triangleq&  \left\lbrace  \alpha_{j}-D(P \| Q), j \in \mathcal{I} \right\rbrace  \label{eq:tauv} \\
	\tau &\triangleq & \max \mathcal{B}. \label{eq:taub}
\end{eqnarray}  
We then have the following result.	
\begin{lemma} \label{optim} Let  $0 < \sqrt{\tilde{r}} < \frac{\sqrt{V(P\|Q)}}{\tau}$. Then, the probability distribution $\Gamma^{*}=(\Gamma^{*}_{1},\cdots,\Gamma^{*}_{k} )$ that minimizes $\ell(\Gamma)$  over $\bar{A}_{\chi^{2}}(\tilde{r})$ is given by
	\begin{eqnarray}
		\Gamma^{*}_{i} & = & P_{i} + \frac{\sqrt{\tilde{r}}\left( D(P \| Q)-\alpha_{i}\right) P_{i} }{\sqrt{V(P\|Q)} }; \quad i=1,\cdots,k. \label{eq:minprob}
	\end{eqnarray} 
	Moreover,
	\begin{eqnarray} 
		\min_{\Gamma \in \bar{A}_{\chi^{2}} \left(  \tilde{r}\right)  } \ell(\Gamma)   =- \sqrt{V(P\|Q) \tilde{r}}. \label{eq:min1}
	\end{eqnarray} 
\end{lemma}  
\begin{IEEEproof}
	See Appendix~\ref{lemmaoptim}.
\end{IEEEproof}
\begin{remark} \label{remmin} Since  $\Gamma^{*}$ and $P$ are probability distributions, by taking summation on both sides of \eqref{eq:minprob}, we obtain
	\begin{eqnarray}
		\sum_{i=1}^{k}  (D(P\|Q)-\alpha_{i})P_{i}=0. \label{eq:z3a}
	\end{eqnarray}
	By the assumption of Theorem~\ref{divergence}, we have that $P \neq Q $. Hence, $D(P\|Q)-\alpha_{i}$ cannot be zero for all $i=1,\cdots,k$. This implies that there exist two distinct indices $l, l' \in \{1,\cdots,k\}$,  such that $D(P\|Q)-\alpha_{l} < 0$ and $D(P\|Q)-\alpha_{l'} > 0$. 
	It follows that $l \in \mathcal{I}$, which in turn yields that  $\mathcal{B}$ is a non-empty set.   
\end{remark}
Let
\begin{eqnarray}
	r'_{n}=\frac{r_{n}}{\eta}+ \frac{M}{ \eta n^{3/2}}. \label{eq:rvalueprime}
\end{eqnarray}
Since $r'_{n}$ is vanishing as $n\rightarrow \infty$, we can choose $\ddot{N}$ such that, for $n \geq \ddot{N}$,  $r'_{n}$  satisfies $ \sqrt{r'_{n}} < \frac{\sqrt{V(P\|Q)}}{\tau}$ with $\tau$ defined in  \eqref{eq:taub}.
Then, \eqref{eq:min} and \eqref{eq:min1} yield that, for $\tilde{P} \in  \mathcal{P}_{n} \cap A_{\tilde{D}} ( r_{n})$ and  $n \geq N_{2}\triangleq\max\{ N_{1}, \ddot{N}_{2}, \ddot{N}\}$, 
\begin{eqnarray}
	\ell(\tilde{P}) &\geq & \min_{\bar{A}_{\chi^{2}}(r'_{n}) } \ell(\Gamma)  \nonumber  \\
	& =&- \sqrt{V(P\|Q) r'_{n}}. \label{eq:minvalue}
\end{eqnarray}
It follows that, for any $\tilde{P}  \in  \mathcal{P}_{n} \cap A_{\tilde{D}} ( r_{n})$,  \eqref{eq:tylq4} can be further lower-bounded as 
\begin{eqnarray}
	D (\tilde{P}\|  Q) 
	& \geq &  D(P \| Q) - \sqrt{V(P\|Q) r'_{n}}- \frac{ M_{2} }{n^{3/2}}, \quad n \geq N_{2}. 
	\label{eq:tylq4a}
\end{eqnarray}
Applying \eqref{eq:tylq4a} to \eqref{eq:tyerr}, the type-II error can be upper-bounded as
\begin{eqnarray}
	\beta_{n}(T_{n}^{\tilde{D}}(r_{n})) & \leq&   \sum_{\tilde{P} \in \mathcal{P}_{n} \cap A_{\tilde{D}} (r_{n})}  \exp \{-n D(\tilde{P} \| Q) \} \nonumber  \\
	& \leq & \sum_{\tilde{P} \in \mathcal{P}_{n} \cap A_{\tilde{D}} (r_{n})}  \exp \left\lbrace  -nD(P \| Q) + n \sqrt{V(P\|Q) r'_{n}}  +  \frac{M_{2}}{\sqrt{n}} \right\rbrace  \nonumber  \\
	& \leq & (n+1)^{\mid \mathcal{Z}\mid}  \exp \left\lbrace  -nD(P \| Q) + n \sqrt{V(P\|Q) r'_{n}}  +  \frac{M_{2}}{\sqrt{n}} \right\rbrace,\quad n \geq N_{2}  \label{eq:tyerror}
\end{eqnarray}
where the last inequality follows since $| \mathcal{P}_{n} | \leq (n+1)^{| \mathcal{Z}|}$ \cite[Th. 11.1.1]{ATCB}. By taking logarithms on both sides of \eqref{eq:tyerror}, we obtain
\begin{eqnarray}
	\ln \beta_{n}(T_{n}^{\tilde{D}}(r_{n})) \leq | \mathcal{Z}| \ln (n+1) -nD(P \| Q) + n \sqrt{V(P\|Q) r'_{n}}  +  \frac{M_{2}}{\sqrt{n}}, \quad  n \geq N_{2}. \label{eq:sub} 
\end{eqnarray}
By \eqref{eq:rvalue} and \eqref{eq:rvalueprime}, 
\begin{eqnarray}
	\sqrt{ r'_{n}} & = & \sqrt{\frac{r_{n}}{\eta}+ \frac{M}{ \eta n^{3/2}}} \nonumber  \\
	&=& \sqrt{\frac{1}{n}Q_{\chi^{2},k-1}^{-1} \left( \epsilon-M_{0} \delta_{n} \right) + \frac{M }{\eta n^{3/2}} } \nonumber  \\
	&= & \frac{1}{\sqrt{n}} \sqrt{Q_{\chi^{2},k-1}^{-1} \left( \epsilon-M_{0} \delta_{n} \right) + \frac{M}{ \eta \sqrt{n}} }.  \label{eq:sqr}
\end{eqnarray}
A Taylor-series expansion of $ \sqrt{x+ \delta} $ around $x>0$ yields then that
\begin{eqnarray}
	\sqrt{ r'_{n}} &= & \frac{1}{\sqrt{n}} \sqrt{Q_{\chi^{2},k-1}^{-1} \left( \epsilon-M_{0} \delta_{n} \right) + \frac{M}{\eta \sqrt{n}} } \nonumber  \\
	&\leq& \frac{1}{\sqrt{n}} \left[  \sqrt{Q_{\chi^{2},k-1}^{-1} \left( \epsilon-M_{0} \delta_{n} \right)} +  \frac{M }{ 2 \eta \sqrt{n} \sqrt{Q_{\chi^{2},k-1}^{-1} \left( \epsilon-M_{0} \delta_{n} \right)}}   \right] \nonumber \\
	&\leq& \frac{1}{\sqrt{n}}   \sqrt{Q_{\chi^{2},k-1}^{-1} \left( \epsilon-M_{0} \delta_{n} \right)}  +  \frac{M}{ 2 \eta n \sqrt{Q_{\chi^{2},k-1}^{-1} ( \epsilon) }}    \label{eq:sub1}
\end{eqnarray}
where the last inequality holds since the function $Q_{\chi^{2},k-1}^{-1} (\cdot)$ is monotonically decreasing. It follows that  \eqref{eq:sub} can be further upper-bounded as
\begin{eqnarray}
	\ln \beta_{n}(T_{n}^{\tilde{D}}(r_{n})) & \leq&  | \mathcal{Z}| \ln (n+1) -nD(P \| Q) + n \sqrt{V(P\|Q) r'_{n}}  +  \frac{M_{2}}{\sqrt{n}}  \nonumber  \\
	& \leq &    |\mathcal{Z}| \ln (n+1) -n  D(P \| Q)+ \sqrt{nV(P\|Q)} \sqrt{Q_{\chi^{2},k-1}^{-1} \left( \epsilon-M_{0} \delta_{n} \right)}  \nonumber \\
	&& \hspace{20mm} +  \frac{\sqrt{V(P\|Q)}M}{ 2 \eta \sqrt{Q_{\chi^{2},k-1}^{-1} ( \epsilon) }}   +\frac{M_{2}}{\sqrt{n}}, \quad n \geq N_{2}. \label{eq:sub4}
\end{eqnarray}
Using a Taylor-series expansion of   $\sqrt{Q_{\chi^{2},k-1}^{-1}(\epsilon-M_0\delta_n)}$ around $\epsilon$, we obtain
\begin{eqnarray}
	\sqrt{Q_{\chi^{2},k-1}^{-1} \left( \epsilon-M_{0} \delta_{n} \right)}
	&=& \sqrt{Q_{\chi^{2},k-1}^{-1} ( \epsilon) } +O\left( \delta_{n} \right). \label{eq:ts} 
\end{eqnarray}
Consequently,  \eqref{eq:sub4} can be expanded as
\begin{eqnarray*}
	\ln \beta_{n}(T_{n}^{\tilde{D}}(r_{n})) & \leq & 
	|\mathcal{Z}| \ln (n+1) -nD(P \| Q)  + \sqrt{nV(P\|Q)} \sqrt{Q_{\chi^{2},k-1}^{-1} ( \epsilon) } +O(\sqrt{n} \delta_{n})  \nonumber \\
	& &  \hspace{5mm} + \frac{\sqrt{V(P\|Q)}M}{ 2 \eta \sqrt{Q_{\chi^{2},k-1}^{-1} ( \epsilon) }}   +\frac{M_{2}}{\sqrt{n}} \nonumber \\
	& = & 
	-nD(P \| Q)  + \sqrt{nV(P\|Q)} \sqrt{Q_{\chi^{2},k-1}^{-1} ( \epsilon) } +O( \max\{ \delta_{n} \sqrt{n}, \ln n \}).
\end{eqnarray*}
This proves Part 1 of Theorem~\ref{divergence}.  
\subsection{ Proof of  Part 2 of Theorem \ref{divergence}}
Define for $0 <\epsilon <1$ and $n \in \Bbb N$, 
\begin{eqnarray}
	R_{n}^{\epsilon} \triangleq \left\lbrace  r > 0 \; \left. \right| \; P^{n} \left( \tilde{D}(t_{Z^{n}} \| P) \geq r \right)  \leq \epsilon \right\rbrace \label{eq:rnset}
\end{eqnarray} 
\begin{eqnarray}
	r^{\epsilon}_{n} \triangleq \inf R_{n}^{\epsilon}. \label{eq:rn}
\end{eqnarray} 
\begin{lemma} \label{spectrum} For $M_{0}$ and $\delta_{n}$ given in \eqref{eq:ratekl1}, we have
	\begin{eqnarray}
		r^{\epsilon}_{n} =\frac{\eta}{n} Q^{-1}_{\chi^{2},k-1}(\epsilon) +  O \left( \frac{\delta_{n}}{n}  \right). \label{eq:sec6o}
	\end{eqnarray}
\end{lemma}
\begin{proof}
	For any $r > 0$, we can write 
	\begin{eqnarray*}
		P^{n} \left( \tilde{D}(t_{Z^{n}} \| P)   \geq r \right) =P^{n} \left( \frac{n}{\eta}  \tilde{D}(t_{Z^{n}} \| P)   \geq \frac{nr}{\eta}  \right).
	\end{eqnarray*}
	Then,  it follows  from \eqref{eq:ratekl1} that, for $n \geq N_{0}$ and any $r>0$,
	\begin{eqnarray}
		Q_{\chi^{2},k-1} \left(  \frac{nr}{\eta}  \right)- M_{0} \delta_{n} \leq P^{n}(\tilde{D}(t_{Z^{n}} \| P)   \geq r) \leq Q_{\chi^{2},k-1} \left( \frac{nr}{\eta} \right)+ M_{0} \delta_{n}.  \label{eq:spec1}
	\end{eqnarray}
	Let
	\begin{eqnarray*}
\dot{r}_{n}= \frac{\eta}{n} Q^{-1}_{\chi^{2},k-1} \left( \epsilon- M_{0} \delta_{n} \right).
	\end{eqnarray*} 
	Substituting the value of $\dot{r}_{n}$ in \eqref{eq:spec1}, we get for $n \geq N_{0}$,
	\begin{eqnarray*}
		P^{n}(\tilde{D}(t_{Z^{n}} \| P)   \geq \dot{r}_{n})  &\leq  \epsilon. \label{eq:spec1rw1}
	\end{eqnarray*}
	Thus,  by the definition of  $r_{n}^{\epsilon}$
	\begin{eqnarray}
		r^{\epsilon}_{n}  &\leq&  \dot{r}_{n} \nonumber \\
		& =& \frac{\eta}{n} Q^{-1}_{\chi^{2},k-1} \left( \epsilon - M_{0} \delta_{n} \right), \quad n \geq N_{0}. \label{eq:rhs}
	\end{eqnarray}
	Furthermore, \eqref{eq:spec1rw1}  implies that $r_{n}^{\epsilon}$ exists. Indeed, the set  $R_{n}^{\epsilon}$ is lower-bounded by $0$. Then, it follows from \eqref{eq:spec1rw1} that, for $n \geq N_{0}$, we have $\dot{r}_{n} \in R_{n}^{\epsilon}$.  This implies that  $R_{n}^{\epsilon}$  is a non-empty subset of the real numbers $\Bbb R$ that has a lower bound.  Hence, by the completeness property of $\Bbb R$,  $r_{n}^{\epsilon}$ exists.  
	
	We next derive a lower bound on $r_n^{\epsilon}$. Let
	\begin{eqnarray*}
		\ddot{r}_{n}= \frac{\eta}{n} Q^{-1}_{\chi^{2},k-1} \left( \epsilon + 2M_{0} \delta_{n} \right). \label{eq:rtilde}
	\end{eqnarray*} 
	Substituting the value of $	\ddot{r}_{n}$ in \eqref{eq:spec1}, we get for $n \geq N_{0}$,
	\begin{eqnarray}
		P^{n}(\tilde{D}(t_{Z^{n}} \| P)   \geq 	\ddot{r}_{n})   \geq   \epsilon + M_{0} \delta_{n} 
		>  \epsilon.  \label{eq:spec1rw1a}
	\end{eqnarray}
	The above equation implies that $	\ddot{r}_{n}$ does not belong to the set $R_{n}^{\epsilon} $.   It follows from the property of the infimum that 
	\begin{eqnarray}
		r^{\epsilon}_{n}  & \geq &  	\ddot{r}_{n} \nonumber \\
		& =& \frac{\eta}{n} Q^{-1}_{\chi^{2},k-1} \left( \epsilon + 2M_{0} \delta_{n} \right), \quad n \geq N_{0}. \label{eq:rhsa}
	\end{eqnarray}
	Thus, from \eqref{eq:rhs}  and \eqref{eq:rhsa} we obtain
	\begin{eqnarray}
		\frac{\eta}{n} Q^{-1}_{\chi^{2},k-1} \left( \epsilon+ 2M_{0} \delta_{n} \right) \leq r_{n}^{\epsilon} \leq  \frac{\eta}{n} Q^{-1}_{\chi^{2},k-1} \left( \epsilon - M_{0} \delta_{n} \right), \quad n \geq N_{0}. \label{eq:rnepsilonin}
	\end{eqnarray}
	Taylor-series expansions of $Q^{-1}_{\chi^{2},k-1} \left( \epsilon+ 2M_{0} \delta_{n} \right) $ and $Q^{-1}_{\chi^{2},k-1} \left( \epsilon - M_{0} \delta_{n} \right)$ around $\epsilon$ yield then that 
	\begin{eqnarray}
		r^{\epsilon}_{n} =\frac{\eta}{n} Q^{-1}_{\chi^{2},k-1}(\epsilon) +  O \left( \frac{\delta_{n}}{n}  \right). \label{eq:sec6}
	\end{eqnarray}
\end{proof}

For any $T \in A_{\tilde{D}}(r^{\epsilon}_{n})$ with $r^{\epsilon}_{n}$ in \eqref{eq:sec6}, it follows from  \eqref{eq:euorder} that 
\begin{eqnarray}
	\|\mathbf{T}-\mathbf{P}\|_{2}= O\left( \frac{1}{\sqrt{n}}\right).
\end{eqnarray}
Then, \eqref{eq:tylpfd} and \eqref{eq:tylq} can be written as 
\begin{eqnarray}
	\tilde{D}(T \| P) &=&   \eta d_{\chi^{2}} (T, P) + O\left( \frac{1}{n^{3/2}} \right) \label{eq:typorderc} \\  
	D (T \|  Q) &=&  D(P \| Q) + \sum_{i=1}^{k} (T_{i} -P_{i} ) \ln    \left( \frac{P_{i}}{ Q_{i}} \right)+  \frac{1}{2} d_{\chi^{2}} (T, P) + O\left( \frac{1}{n^{3/2}} \right).  \label{eq:tyqorderc}
\end{eqnarray}
The first equation  \eqref{eq:typorderc} implies that there exist $\bar{M}'_{0}>0$ and $\bar{N}'_{0} \in \Bbb N$ such that, for all $ n \geq  \bar{N}'_{0} $ and for $T \in A_{\tilde{D}}(r^{\epsilon}_{n})$,
\begin{eqnarray}
	|\tilde{D}(T \| P) -  \eta d_{\chi^{2}} (T, P) | < \frac{ \bar{M}'_{0} }{n^{3/2}}. \label{eq:tylp2w}
\end{eqnarray}
If  $T \in \bar{A}_{\chi^{2}} \left( \frac{r^{\epsilon}_{n}}{\eta}-\frac{\bar{M}'_{0}}{\eta n^{3/2}}\right)$, then we have $d_{\chi^{2}}(T,P)\leq \frac{r^{\epsilon}_{n}}{\eta}-\frac{\bar{M}'_{0}}{\eta n^{3/2}}$. 
Furthermore, by the definition of the $\chi^2$-divergence in \eqref{eq:chisquare} and the fact that $P_{i} <1,\; i=1,\cdots,k$, we have that $d_{\chi^{2}} (T, P) \geq  \|T-P\|_{2}^{2}$.  Since $r^{\epsilon}_{n}=O\left( \frac{1}{n}\right)$,  these two inequalities imply that  $\|T-P\|_{2}= O\left( \frac{1}{\sqrt{n}}\right) $ and,  hence, also   $	\|\mathbf{T}-\mathbf{P}\|_{2}= O\left( \frac{1}{\sqrt{n}}\right)$ for every $T \in \bar{A}_{\chi^{2}} \left( \frac{r^{\epsilon}_{n}}{\eta}-\frac{\bar{M}'_{0}}{\eta n^{3/2}}\right)$.
It follows  from \eqref{eq:tylpfd} that  there exist $\bar{M}'_{1}>0$ and $\bar{N}'_{1} \in \Bbb N$ such that, for all $ n \geq  \bar{N}'_{1}$ and  $T \in \bar{A}_{\chi^{2}} \left( \frac{r^{\epsilon}_{n}}{\eta}-\frac{\bar{M}'_{0}}{\eta n^{3/2}}\right)$, 
\begin{eqnarray}
	|\tilde{D}(T \| P) -  \eta d_{\chi^{2}} (T, P) | < \frac{ \bar{M}'_{1} }{n^{3/2}}. \label{eq:tylp2bw}
\end{eqnarray}
Then we have the following  lemma.
\begin{lemma} \label{lemmaballw} For $ n \geq N'_{1}\triangleq\max \{ \bar{N}'_{0}, \bar{N}'_{1}\}$, we have
	\begin{eqnarray}
		\bar{A}_{\chi^{2}} \left( \frac{r^{\epsilon}_{n}}{\eta}-\frac{M'_{1}}{\eta n^{3/2}}\right) \subseteq A_{\tilde{D}}(r^{\epsilon}_{n}) \label{eq:klchiw}
	\end{eqnarray}
	where $M'_{1}\triangleq\max \{ \bar{M}'_{0}, \bar{M}'_{1}\}$ and  $r^{\epsilon}_{n}$ satisfies \eqref{eq:sec6}.
\end{lemma}
\begin{IEEEproof} 
	Since $M'_{1}=\max \{ \bar{M}'_{0}, \bar{M}'_{1}\}$,  we have $\bar{M}'_{0}\leq M'_{1}$. This implies that 
	\begin{eqnarray}
		\bar{A}_{\chi^{2}} \left( \frac{r^{\epsilon}_{n}}{\eta}-\frac{M'_{1}}{\eta n^{3/2}}\right)  \subseteq 	\bar{A}_{\chi^{2}} \left( \frac{r^{\epsilon}_{n}}{\eta}-\frac{\bar{M}'_{0}}{\eta n^{3/2}} \right).
	\end{eqnarray}
	Then, using  \eqref{eq:tylp2bw}, we have for $T \in 	\bar{A}_{\chi^{2}} \left( \frac{r^{\epsilon}_{n}}{\eta}-\frac{M'_{1}}{\eta n^{3/2}}\right)$ and for $n \geq N'_{1}=\max \{ \bar{N}'_{0}, \bar{N}'_{1}\}$ that
	\begin{eqnarray*}
		\tilde{D}(T \| P)  &< &   \eta d_{\chi^{2}} (T, P) + \frac{ \bar{M}'_{1} }{n^{3/2}} \\
		&\leq & \eta  \left(\frac{r^{\epsilon}_{n}}{\eta}-\frac{M'_{1}}{\eta n^{3/2}}\right) +\frac{ \bar{M}'_{1} }{n^{3/2}}\\
		&\leq& r^{\epsilon}_{n}.
	\end{eqnarray*}
	Hence, we obtain $T \in A_{\tilde{D}}(r^{\epsilon}_{n})$.
\end{IEEEproof}

From \eqref{eq:tyqorderc},  we obtain that, for $T \in A_{\tilde{D}}(r^{\epsilon}_{n})$, there exist $M_{2}'>0$ and $N_{2}' \in \Bbb N$ such that
\begin{eqnarray}
	|D (T \|  Q) -  D(P \| Q) - \sum_{i=1}^{k} (T_{i} -P_{i} ) \ln    \left( \frac{P_{i}}{ Q_{i}} \right)-  \frac{1}{2} d_{\chi^{2}} (T, P) | \leq \frac{ M_{2}' }{n^{3/2}},\quad n \geq  N_{2}'.  \label{eq:tylq2c}
\end{eqnarray}
If $r < r^{\epsilon}_{n}$, then, by definition of $r^{\epsilon}_{n}$, the type-I error  $\alpha_{n}(T^{\tilde{D}})$ exceeds $ \epsilon$. Hence such a threshold violates 
\eqref{eq:typeepsilon}. We thus assume without loss of optimality that $r\geq r^{\epsilon}_{n}$. In this case, 
\begin{eqnarray}
	\beta_{n}(T_{n}^{\tilde{D}}(r)) \triangleq	Q^{n} \left( \tilde{D}(t_{Z^{n}} \| P)   < r \right)  \geq Q^{n} \left( \tilde{D}(t_{Z^{n}} \| P)   < r^{\epsilon}_{n} \right)  . \label{eq:rgeq}
\end{eqnarray}
Together with Lemma~\ref{lemmaballw}, this implies that, for $r^{\epsilon}_{n}$ and  $n \geq N_{1}'$,  the type-II error of the divergence test $T_{n}^{\tilde{D}}(r)$ can be lower-bounded as
\begin{eqnarray}
	\beta_{n}(T_{n}^{\tilde{D}}(r))  &\geq&   Q^{n}(\tilde{D}(t_{Z^{n}} \| P)   < r^{\epsilon}_{n})  \nonumber \\
	&= &  \sum_{ z^{n}:\; t_{z^{n}} \in A_{\tilde{D}}(r^{\epsilon}_{n})} Q^{n}(z^{n}) \nonumber \\
	&=&\sum_{\tilde{P} \in \mathcal{P}_{n} \cap A_{\tilde{D}}(r^{\epsilon}_{n})}  Q^{n}(T(\tilde{P})) \nonumber  \\
	&\geq&  \sum_{\tilde{P} \in \mathcal{P}_{n} \cap \bar{A}_{\chi^{2}} \left(  \bar{r}_{n}\right) }  Q^{n}(T(\tilde{P})) \label{eq:sec3c}
\end{eqnarray}
where  
\begin{equation}
	\bar{r}_{n}=\frac{r^{\epsilon}_{n}}{\eta}-\frac{M'_{1}}{\eta n^{3/2}}. \label{eq:rbar}
\end{equation}
By  \eqref{eq:sec6},  $\bar{r}_{n}$ vanishes as $n \to \infty$. We can therefore find an  $\tilde{N}'$ such that, for $n \geq \tilde{N}'$, we have that $ \sqrt{\bar{r}_{n}} < \frac{\sqrt{V(P\|Q)}}{\tau}$ with $\tau$ defined in \eqref{eq:taub}.  To lower-bound the type-II error further, let us consider the following lemma. 
\begin{lemma}\label{type}  Consider the minimizing  probability distribution $\Gamma^{*} $ given in  \eqref{eq:minprob} in Lemma~\ref{optim}, namely
	\begin{eqnarray}
		\Gamma^{*}_{i} &=& P_{i} + \frac{\sqrt{\bar{r}_{n}} (D(P\|Q)-\alpha_{i})P_{i}}{\sqrt{V(P\|Q)} }; \quad i=1,\cdots,k \label{eq:z1}
	\end{eqnarray}
	with $\bar{r}_{n}$ given in \eqref{eq:rbar}, $\alpha_{i}$ given  in \eqref{eq:alphai}, and $n \geq \tilde{N}'$. Then, there exist $\tilde{N} \in \Bbb N$ and  a type distribution $T^{*}_{n}=(T^{*}_{n}(a_{1}), \cdots, T^{*}_{n}(a_{k}) ) \in \mathcal{P}_{n} \cap \bar{A}_{\chi^{2}} \left(  \bar{r}_{n}\right) $ such that  
	\begin{eqnarray}
		| n \ell(	\Gamma^{*}) -n\ell(T^{*}_{n})| \leq \kappa, \quad n \geq \tilde{ N} \label{eq:hfunb}
	\end{eqnarray}
	for some constant $\kappa >0$.
\end{lemma}
\begin{IEEEproof} See Appendix~\ref{lemmatype}.
\end{IEEEproof}

From \eqref{eq:sec3c}, the type-II error of the divergence test  can be lower-bounded as 
\begin{eqnarray}
	\beta_{n}(T_{n}^{\tilde{D}}(r))  &\geq&  \sum_{\tilde{P} \in \mathcal{P}_{n} \cap \bar{A}_{\chi^{2}} \left(  \bar{r}_{n}\right) }  Q^{n}(T(\tilde{P})) \nonumber \\
	& \geq &  Q^{n}(T(T^{*}_{n})) \nonumber \\
	& \geq&  \frac{1}{(n+1)^{|\mathcal{Z}|}} \exp \{-n D(T^{*}_{n} \| Q) \}, \quad   n \geq \tilde{N}_{1} \label{eq:prtyclc} 
\end{eqnarray}
where $\tilde{N}_{1} \triangleq \max \{N_{1}', N_{2}',  \tilde{N}\}$ and
$T^{*}_{n}  \in \mathcal{P}_{n} \cap \bar{A}_{\chi^{2}} \left(  \bar{r}_{n} \right) $ is a type distribution satisfying \eqref{eq:hfunb}. The last inequality in 
\eqref{eq:prtyclc} follows from \cite[Th. 11.1.4]{ATCB}.

We next upper-bound $D (T^{*}_{n} \|  Q) $. To this end, we use \eqref{eq:tylq2c} and \eqref{eq:hfunb}, and we recall that  $\ell(	\Gamma^{*}) =-\sqrt{V(P,Q)\bar{r}_{n}} $, to obtain
\begin{eqnarray}
	n D (T^{*}_{n} \|  Q)  & \leq &  nD(P \| Q) + n\ell(T^{*}_{n}) + \frac{1}{2} n d_{\chi^{2}} (T^{*}_{n}, P) + \frac{M'_{2}}{\sqrt{n}} \nonumber \\
	& \leq &  nD(P \| Q) + n \ell(\Gamma^{*}) +\kappa  + \frac{1}{2} n d_{\chi^{2}} (T^{*}_{n}, P) + \frac{M'_{2}}{\sqrt{n}}  \nonumber \\
	& \leq & nD(P \| Q) - n \sqrt{V(P,Q)\bar{r}_{n}} + \kappa + \frac{n\bar{r}_{n}}{2} + \frac{M'_{2}}{\sqrt{n}}, \quad n \geq \tilde{N}_{1} \label{eq:sec2}
\end{eqnarray}
where the last inequality follows because $T_n^* \in \bar{A}_{\chi^{2}} \left(  \bar{r}_{n} \right)$.
Substituting  \eqref{eq:sec2} in  \eqref{eq:prtyclc}, we get
\begin{eqnarray}
	\beta_{n}(T_{n}^{\tilde{D}}(r))  
	& \geq &\frac{1}{(n+1)^{|\mathcal{Z}|}} \exp\left[-nD(P \| Q) + n \sqrt{V(P\|Q)\bar{r}_{n}} - \kappa - \frac{n\bar{r}_{n}}{2} - \frac{M'_{2}}{\sqrt{n}}   \right], \quad n \geq \tilde{N}_{1}. 
\end{eqnarray}
Taking logarithm on both sides, we obtain
\begin{eqnarray}
	\ln \beta_{n}(T_{n}^{\tilde{D}}(r))  & \geq &-|\mathcal{Z}| \ln (n+1)  -nD(P \| Q) + n \sqrt{V(P\|Q)\bar{r}_{n}} - \kappa - \frac{n\bar{r}_{n}}{2} - \frac{M'_{2}}{\sqrt{n}}.  \label{eq:sec5}
\end{eqnarray}
Note that, by \eqref{eq:sec6} and \eqref{eq:rbar},  
\begin{eqnarray}
	\bar{r}_{n}
	&= &\frac{1}{n} Q^{-1}_{\chi^{2},k-1}(\epsilon) +  O \left( \frac{\delta_{n}}{n}  \right) -\frac{M'_{1}}{\eta n^{3/2}}.  \label{eq:sec8c}
\end{eqnarray}
Depending on whether the last two terms are of order  $O \left( \frac{1}{n^{3/2}}  \right)$ or of order $O \left( \frac{\delta_{n}}{n} \right)$, we either have 
\begin{eqnarray}
	\sqrt{\bar{r}_{n}} &=&
	\sqrt{\frac{1}{n} Q^{-1}_{\chi^{2},K-1}(\epsilon) +  O \left( \frac{1}{n^{3/2}}\right)  }  \nonumber \\
	& = & \frac{1}{\sqrt{n}}  \sqrt{  Q^{-1}_{\chi^{2},k-1}(\epsilon)} +  O \left( \frac{1}{n} \right)   \label{eq:sec8}
\end{eqnarray}
or
\begin{eqnarray}
	\sqrt{\bar{r}_{n}} &=&
	\sqrt{\frac{1}{n} Q^{-1}_{\chi^{2}}(\epsilon) +  O \left( \frac{\delta_{n}}{n}  \right) } \nonumber \\
	& = & \frac{1}{\sqrt{n}}  \sqrt{  Q^{-1}_{\chi^{2}}(\epsilon)} + O \left( \frac{\delta_{n}}{\sqrt{n}} \right).    \label{eq:sec8a}
\end{eqnarray}
It follows that \eqref{eq:sec5} can be written as 
\begin{eqnarray}
	-	\ln \beta_{n}(T_{n}^{\tilde{D}}(r))&\leq & 
	nD(P \| Q) -  \sqrt{nV(P\|Q) Q^{-1}_{\chi^{2}}(\epsilon) } +O( \max\{ \delta_{n} \sqrt{n}, \ln n \}).
\end{eqnarray}
This proves Part 2 of Theorem~\ref{divergence}. 

\section{Conclusions} \label{sec:conclusion}

 For the divergence test and for a large class of divergences, we established the first-order and the second-order terms of the type-II error given that the type-I error is upper-bounded by a given value. The divergence test does not require  the knowledge of the distribution of the alternative hypothesis, hence, it is suitable for problems where we only have access to the distribution of the null hypothesis. The class of divergences considered in this paper  includes well-known divergences such as the KL divergence and the $\alpha$-divergence, and our divergence test specializes to the Hoeffding test if the chosen divergence is the KL divergence. It is well-known that the Hoeffding test is first-order optimal in the sense that its first-order term is equal to the first-order term of the Neyman-Pearson test. However, our results demonstrate that the second-order term of the divergence test for the class of divergences considered in this paper, and therefore also of the Hoeffding test, is strictly smaller than the second-order term of the optimal Neyman-Pearson test. The question whether there exists a test  that only requires the knowledge of the distribution of the null hypothesis and that achieves the second-order term of the Hoeffding test for the some $Q$'s while it achieves a higher second-order term for the remaining  $Q$'s  is yet to be  explored.

\appendices

\section{Proof of Lemma \ref{ctaylor}} \label{claimtaylor}
For the probability distribution $T=(T_{1},\cdots, T_{k})^{T}$ and $Q \in \mathcal{P}(\mathcal{Z})$,  the function $f_{Q}(T) \triangleq D(T \| P)$  is a smooth function   of $\mathbf{T}=(T_{1}, \cdots, T_{k-1})^{T}$ given by
\begin{eqnarray}
	f_{Q}(\mathbf{T}) 
	&=& \sum_{i=1}^{k-1} T_{i}   \ln    \left( \frac{T_{i}}{ Q_{i}} \right) + \left( 1-\sum_{i=1}^{k-1} T_{i}\right)   \ln    \left( \frac{(1-\sum_{i=1}^{k-1} T_{i})}{ Q_{k}}\right). \label{eq:ts0}
\end{eqnarray}
The  partial derivatives of  $f_{Q}(\cdot)$ with respect to $T_{i}, i=1,\cdots, k-1$ up to third orders  are given by
\begin{eqnarray}
	\frac{\partial f_{Q}(\mathbf{T})}{\partial T_{i}}  &=& \ln \left( \frac{T_{i}}{ Q_{i}} \right) -\ln \left( \frac{T_{k}}{ Q_{k}} \right)  \label{eq:ts1}
\end{eqnarray}
\begin{eqnarray}
	\frac{\partial^{2} f_{Q}(\mathbf{T})}{\partial T_{i} T_{j}}  &=& 
	\begin{cases}
		\frac{1}{T_{i}} +\frac{1}{T_{k}},  \quad	& \text{if}  \quad i=j\\
		\frac{1}{T_{k}},  \quad 	&   \text{if}  \quad  i \neq j
	\end{cases}
\label{eq:ts2}
\end{eqnarray}
\begin{eqnarray}
	\frac{\partial^{3} f_{Q}(\mathbf{T})}{\partial T_{i} T_{j} T_{l}}  &=& 
	\begin{cases}
		\frac{-1}{T_{i}^{2}} +\frac{1}{T_{k}^{2}},  \quad	& \text{if}  \quad i=j=l\\
		\frac{1}{T_{k}^{2}},  \quad 	&     \text{otherwise}.
	\end{cases}
	 \label{eq:ts3}
\end{eqnarray}
Consequently,  the second-order Taylor approximation of  $f_{Q}(\mathbf{T})$  around $P$ is given by
\begin{eqnarray}
	f_{Q}(\mathbf{T}) &=&  f_{Q}(\mathbf{P})+ \sum_{i=1}^{k-1}\frac{\partial f_{Q}(\mathbf{P})}{\partial T_{i}} (T_{i} -P_{i}) +  \frac{1}{2}\sum_{i,j=1}^{k-1} \frac{\partial^{2} f_{Q}(\mathbf{P})}{\partial T_{i} T_{j}} (T_{i} -P_{i})(T_{j} -P_{j}) +R_{3}(\mathbf{T})
	\label{eq:ts4}
\end{eqnarray}
where the remainder term $R_{3}(\mathbf{T})$  in  Lagrange's form is given  by
\begin{eqnarray}
	R_{3}(\mathbf{T})= \frac{1}{3!} \sum_{i,j,l=1}^{k-1} \frac{\partial^{3} f_{Q}(\Pi)}{\partial T_{i} T_{j} T_{l}}  (T_{i} -P_{i})(T_{j} -P_{j})(T_{l} -P_{l})
	\label{eq:ts5}
\end{eqnarray}
for $\mathbf{\Pi}=\mathbf{P}+t(\mathbf{T} - \mathbf{P})$, $t \in (0,1)$. The absolute value of this term can be upper-bounded by
\begin{eqnarray}
	|R_{3}(\mathbf{T})| &=&  \left| \frac{1}{3!} \sum_{i,j,l=1}^{k-1}  \frac{\partial^{3} f_{Q}(\mathbf{\Pi})}{\partial T_{i} T_{j} T_{l}}  (T_{i} -P_{i})(T_{j} -P_{j})(T_{l} -P_{l}) \right| 
	\notag \\
	&\leq&  \frac{1}{3!} \sum_{i,j,l=1}^{k-1} \left|  \frac{\partial^{3} f_{Q}(\mathbf{\Pi})}{\partial T_{i} T_{j} T_{l}}  \right|  \left| (T_{i} -P_{i})(T_{j} -P_{j})(T_{l} -P_{l}) \right|.
	\label{eq:ts5w2}
\end{eqnarray}
When $\|\mathbf{T} -\mathbf{P} \|_{2} \rightarrow 0$, we have that $P_{i}-\delta<T_{i}<P_{i}+\delta$ for an arbitrary $\delta>0$ and $i=1,\cdots, k-1$. This implies that  $P_{k}-(k-1) \delta <T_{k}<P_{k}+(k-1) \delta$. Choosing  $\delta\triangleq\frac{1}{4(k-1)} \min_{1\leq i \leq k} P_{i}$, it follows that $P_{i}-\delta_{0}>0$ and $P_{k}-(k-1) \delta >0$. Consequently, 
\begin{eqnarray}
	\frac{1}{T_{i}^{2}} < \frac{1}{(P_{i}-\delta)^{2}}, \quad  i=1,\cdots, k-1  \label{eq:lbty}
\end{eqnarray}
and 
\begin{eqnarray}
	\frac{1}{T_{k}^{2}} < \frac{1}{(P_{k}-(k-1) \delta)^{2}}. \label{eq:lbtya}
\end{eqnarray}
Using \eqref{eq:lbty}, \eqref{eq:lbtya}, and   \eqref{eq:ts3}, one can show that, for any $\mathbf{T}$ satisfying $ 0 < \|\mathbf{T} -\mathbf{P}\|_{2} < \delta$,
\begin{eqnarray}
	\left| 	\frac{\partial^{3} f_{Q}(\mathbf{T})}{\partial T_{i} T_{j} T_{l}} \right| &\leq&  C_{1}
\end{eqnarray}
where
\begin{eqnarray}
	C_{1}\triangleq \max_{i}  \frac{1}{(P_{i}-\delta)^{2}} +\frac{1}{(P_{k}-(k-1) \delta)^{2}}.
\end{eqnarray}
Hence, from \eqref{eq:ts5w2},
\begin{eqnarray}
	|R_{3}(\mathbf{T})| &\leq & \frac{1}{3!} C_{1} \sum_{i,j,l=1}^{k-1}  \left| (T_{i} -P_{i})(T_{j} -P_{j})(T_{l} -P_{l}) \right|  \notag \\
	&\leq &	C' \| \mathbf{T} -\mathbf{P}\|_{2}^{3}
	\label{eq:ts5w1e}
\end{eqnarray}
where $C'\triangleq\frac{1}{3!} C_{1}(k-1)^{3}$ and the last inequality follows since $|T_{i} -P_{i}| \leq \|\mathbf{T} -\mathbf{P}\|_{2}$ for all $i=1,\cdots, k-1$. 

Note that $ f_{Q}(P)=D(P \|Q)$. Since $T_{k}-P_{k}=\sum_{i=1}^{k-1}(P_{i}-T_{i})$, we obtain that
\begin{eqnarray}
	\sum_{i=1}^{k-1}\frac{\partial f_{Q}(P)}{\partial T_{i}} (T_{i} -P_{i}) &=& 
	\sum_{i=1}^{k-1}  \left[ \ln \left( \frac{P_{i}}{ Q_{i}} \right) -\ln \left( \frac{P_{k}}{ Q_{k}} \right) \right]  (T_{i} -P_{i})  \notag \\
	&=& \sum_{i=1}^{k} (T_{i} -P_{i}) \ln \left( \frac{P_{i}}{ Q_{i}} \right)   
	\label{eq:ts8}
\end{eqnarray}
and 
\begin{align}
	\sum_{i,j=1}^{k-1} \frac{\partial^{2} f_{Q}(P)}{\partial T_{i} T_{j}} (T_{i} -P_{i})(T_{j} -P_{j}) &=
	\sum_{i=1}^{k-1} \left( \frac{1}{P_{i}} +\frac{1}{P_{k}}\right)  (T_{i} -P_{i})^{2}+  \sum_{ i,j=1,\cdots,k-1,i\neq j}  \frac{(T_{i} -P_{i})(T_{j} -P_{j})}{P_{k}}  \notag   \\
	&=
	\sum_{i=1}^{k-1}  \frac{(T_{i} -P_{i})^{2}}{P_{i}} +\sum_{i,j=1 }^{k-1}  \frac{(T_{i} -P_{i})(T_{j} -P_{j})}{P_{k}}  \notag \\
	&=\sum_{i}^{k}  \frac{(T_{i} -P_{i})^{2}}{P_{i}}.
	\label{eq:ts9}
\end{align}
Thus, by \eqref{eq:ts8}, \eqref{eq:ts9}, and the definition of the $\chi^{2}$-divergence,  \eqref{eq:ts4} can be written as 
\begin{eqnarray}
	f_{Q}(T) &=&  D(P \parallel Q) + \sum_{i=1}^{k} (T_{i} -P_{i} ) \ln    \left( \frac{P_{i}}{ Q_{i}} \right)+  \frac{1}{2} d_{\chi^{2}} (T, P) + O(\| \mathbf{T} -\mathbf{P}\|_{2}^{3})  
\end{eqnarray}
as $\|\mathbf{T} -\mathbf{P}\|_{2} \rightarrow 0$.
\section{Proof of Lemma \ref{optim}}
\label{lemmaoptim}
Let $\Gamma=(\Gamma_{1}, \cdots, \Gamma_{k})$.  Consider the minimization problem
\begin{eqnarray}
	\underset{\Gamma}{\text{Minimize}} \; \;& & \ell(\Gamma)=\sum_{i=1}^{k} ( \Gamma_{i}- P_{i})  \alpha_{i} \\
	\text{subject to }  & &  d_{\chi^{2}} (\Gamma, P)  \leq  \tilde{r} \\
	&& \Gamma_{i} >0, \quad i=1,\cdots,k \\
	&&  \sum_{i=1}^{k}\Gamma_{i} =1.  
\end{eqnarray}
This is equivalent to 
\begin{eqnarray}
	\underset{\Gamma}{\text{Minimize}} \; \;& & \ell(\Gamma)=\sum_{i=1}^{k} ( \Gamma_{i}- P_{i})  \alpha_{i} \\
	\text{subject to }  & &  g_{0}(\Gamma) \triangleq \sum_{i=1}^{k}  \frac{(\Gamma_{i} - P_{i})^{2}}{ P_{i}} - \tilde{r} \leq 0 \\
	&& g_{i}(\Gamma)\triangleq-\Gamma_{i} <0, \quad i=1,\cdots,k \\
	&& h_{1}(\Gamma)\triangleq \sum_{i=1}^{k}\Gamma_{i} -1=0.  
\end{eqnarray}
To find the solution of this constrained optimization problem, we consider the Karush-Kuhn-Tucker (KKT) conditions. Indeed, consider the Lagrangian function
\begin{eqnarray}
	\mathcal{L}(\Gamma,  \lambda, \mu) &=& \ell(\Gamma) + \lambda_{0} g_{0}(\Gamma) +\sum_{i=1}^{k} \lambda_{i} g_{i}(\Gamma) + \mu h_{1}(\Gamma) \notag \\
	& = & \sum_{i=1}^{k} ( \Gamma_{i}- P_{i})  \alpha_{i} + \lambda_{0} \left( \sum_{i=1}^{k}  \frac{(\Gamma_{i} - P_{i})^{2}}{ P_{i}} - \tilde{r}\right)   -\sum_{i=1}^{k} \lambda_{i} \Gamma_{i} + \mu \left( \sum_{i=1}^{k}\Gamma_{i}-1\right) 
\end{eqnarray}
where $\lambda=(\lambda_{0}, \lambda_{1},\cdots,\lambda_{k})$ and $\mu$ are KKT multipliers. 	Let $\Gamma^{*}=(\Gamma_{1}^{*}, \cdots,\Gamma_{k}^{*} )^{T}$ denote minimizing probability distribution. The  KKT conditions for the above problem are as follows:
\begin{align}
\frac{\partial \mathcal{L}}{\partial \Gamma_{i}} &=0,\quad i=1,\cdots,k \label{eq:lm1}\\
	g_{0}(\Gamma^{*}) & \leq0 \label{eq:lm3}\\
	g_{i}(\Gamma^{*})  &< 0, \quad i=1,\cdots,k \label{eq:lm3w}\\
	h_{1}(\Gamma^{*}) & =0 \label{eq:lm9}\\
	 \lambda_{i}  &\geq 0, \quad i=0, 1,\cdots,k \\
	  \lambda_{i} g_{i}(\Gamma^{*}) &= 0, \quad i=0, 1,\cdots,k
\end{align}
which  can be evaluated as 
\begin{align}
	 \alpha_{i}+ 2\lambda_{0} \left(  \frac{\Gamma^{*}_{i}}{P_{i}}-1 \right)  -\lambda_{i}+ \mu & =0,\quad i=1,\cdots,k \label{eq:kkt1}\\
	\sum_{i=1}^{k}  \frac{(\Gamma^{*}_{i} - P_{i})^{2}}{ P_{i}} - \tilde{r} & \leq0 \label{eq:kkt2} \\
	 -\Gamma^{*}_{i}  & <0, \quad i=1,\cdots,k  \label{eq:kkt3}\\
	\sum_{i=1}^{k}\Gamma^{*}_{i}-1 &=0 \label{eq:kkt4}\\
	 \lambda_{i} & \geq 0, \quad i=0, 1,\cdots,k  \label{eq:kkt5}\\
	  \lambda_{0} \left(  \sum_{i=1}^{k}  \frac{(\Gamma^{*}_{i} - P_{i})^{2}}{ P_{i}} - \tilde{r} \right) & = 0  \label{eq:kkt6}\\
	- \lambda_{i} \Gamma^{*}_{i} &=0, \quad i=1,\cdots,k. \label{eq:kkt7}
\end{align}
From \eqref{eq:kkt3} and \eqref{eq:kkt7}, we get that 
\begin{eqnarray}
	\lambda_{i}=0, \quad i=1,\cdots,k. \label{eq:kkti1}
\end{eqnarray}
We next analyze $\lambda_{0}$, considering the cases $\lambda_{0} =0$ and $\lambda_{0} \neq 0$ separately. When $\lambda_{0} =0$,   \eqref{eq:kkt7} and  \eqref{eq:kkt1} yield 
\begin{align}
	& \alpha_{i}  + \mu=0,\quad i=1,\cdots,k. \label{eq:kkt1i2}
\end{align}
Thus, $\ln \frac{P_{i}}{Q_{i}}=-\mu$ for $i=1,\cdots,k$ which implies that $P=Q$. However, we have $P \neq Q $ by the assumption of  Theorem~\ref{divergence}, so 
$\lambda_{0} =0$ is not a  valid solution.

Thus, we have  $\lambda_{0} > 0$. Then, we obtain from  \eqref{eq:kkt1}  and \eqref{eq:kkt6} that 
\begin{eqnarray}
	\Gamma^{*}_{i} & = & P_{i}-  \frac{P_{i}\mu}{2\lambda_{0}} -\frac{P_{i}\alpha_{i}}{2 \lambda_{0}}, \quad i=1,\cdots,k \label{eq:lm5} 
\end{eqnarray}
and 
\begin{eqnarray}
	\sum_{i=1}^{k}  \frac{(\Gamma^{*}_{i} - P_{i})^{2}}{ P_{i}} - \tilde{r}=0. \label{eq:kkti3} 
\end{eqnarray}
Substituting \eqref{eq:lm5}  in \eqref{eq:kkt4} yields 
\begin{eqnarray}
	\frac{-\mu}{2\lambda_{0}}-\frac{1}{2 \lambda_{0}} \sum_{i=1}^{k} P_{i} \alpha_{i}=0
\end{eqnarray}
hence
\begin{eqnarray}
	\mu=-  \sum_{i=1}^{k} P_{i} \alpha_{i}=-D(P \| Q).\label{eq:lm7}
\end{eqnarray}
Using \eqref{eq:lm5} and \eqref{eq:lm7}, \eqref{eq:kkti3} can be written as
\begin{eqnarray}
	\frac{1}{4 \lambda_{0}^{2}} \sum_{i=1}^{k} P_{i} \left( \alpha_{i}-D(P \| Q)\right) ^{2} =\tilde{r}.\label{eq:lm8}
\end{eqnarray}
Since
\begin{eqnarray}
	\sum_{i=1}^{k} P_{i} \left( \alpha_{i}-D(P \| Q)\right) ^{2}=  V(P\| Q)
\end{eqnarray}
it follows that
\begin{eqnarray}
	\lambda_{0} &=&  \frac{\sqrt{V(P\| Q)}}{2\sqrt{\tilde{r}}}. \label{eq:lm4} 
\end{eqnarray}
From \eqref{eq:lm5}, \eqref{eq:lm7},  and \eqref{eq:lm4}, we obtain that
\begin{eqnarray}
	\Gamma^{*}_{i} & = & P_{i} + \frac{\sqrt{\tilde{r}} (D(P \| Q)-\alpha_{i}) P_{i}}{\sqrt{V(P\| Q)} }, \quad i=1,\cdots,k. \label{eq:gamma}
\end{eqnarray}
Note that we have  $P_{i}>0, i=1,\cdots,k$. Hence, if $D(P \| Q)-\alpha_{j} \geq 0$, then it follows from the above equation that  $\Gamma^{*}_{j} >0$. Similarly, if $D(P \| Q)-\alpha_{j} <0$ for some $j \in \{ 1,\cdots,k \}$, then  $j \in \mathcal{I}$ and $\alpha_{j}- D(P \| Q) \in \mathcal{B}$, with $\mathcal{I}$ and $ \mathcal{B}$ defined  in \eqref{eq:tauin} and \eqref{eq:tauv}, respectively. Since  $0<\sqrt{\tilde{r}} < \frac{\sqrt{V(P\|Q)}}{\tau}$, it follows that, for any $j \in \mathcal{I}$, 
\begin{eqnarray}
	\sqrt{\tilde{r}} &<& \frac{\sqrt{V(P\|Q)}}{\tau}  \notag \\
	&	\leq  & \frac{\sqrt{V(P\|Q)}}{\alpha_{j}- D(P \| Q)} \label{eq:rlessvpq}
\end{eqnarray}
where  the last inequality follows since $\alpha_{j}- D(P \| Q) \leq \tau$ for any $j \in \mathcal{I}$ by the definition of $\tau$ in \eqref{eq:taub}. It follows from \eqref{eq:gamma} and  \eqref{eq:rlessvpq} that  $\Gamma^{*}_{j} >0$.

To summarize, the solution to \eqref{eq:kkt1}--\eqref{eq:kkt7} is given by
\begin{eqnarray}
	\lambda_{0} &=& \frac{\sqrt{V(P\| Q)}}{2\sqrt{\tilde{r}}}\\
	\lambda_{i} &=& 0,\quad i=1,\cdots,k\\
	\mu &=&-D(P \| Q)\\
	\Gamma^{*}_{i} & = & P_{i} + \frac{\sqrt{\tilde{r}} (D(P \| Q)-\alpha_{i}) P_{i}}{\sqrt{V(P\| Q)} }, \quad i=1,\cdots,k.
\end{eqnarray}
Since the objective function  $\ell(\cdot)$ of the minimization problem  and the inequality constraints $g_{i}(\cdot), i=0,1,\cdots,k$ are convex functions and the equality constraint $h_{1}(\cdot)$ is affine,  the above solution is optimal. Thus,
the minimizing probability distribution $\Gamma^{*}=(\Gamma^{*}_{1},\cdots,\Gamma^{*}_{k} )$ is given by 
\begin{eqnarray}
	\Gamma^{*}_{i} & = & P_{i} + \frac{\sqrt{\tilde{r}}\left( D(P \| Q)-\alpha_{i}\right) P_{i} }{\sqrt{V(P\|Q)} }, \quad i=1,\cdots,k
\end{eqnarray}
and the minimum value is 
\begin{eqnarray}
	\min_{\Gamma \in \bar{A}_{\chi^{2}}(\tilde{r}) } \ell(\Gamma)   =  \ell(\Gamma^{*}) =- \sqrt{V(P\| Q) \tilde{r}}.
\end{eqnarray}

\section{Proof of Lemma \ref{type}} \label{lemmatype}
Note that, by  \eqref{eq:z1}, 
the probability distribution $\Gamma^{*}=(\Gamma^{*}_{1}, \cdots, \Gamma^{*}_{k})^{T} $ is given by
\begin{eqnarray}
\Gamma^{*}_{i} &=& P_{i} + \frac{\sqrt{\bar{r}} (D(P\|Q)-\alpha_{i})P_{i}}{\sqrt{V(P\|Q)} }, \quad i=1,\cdots,k, \quad n \geq \tilde{N}'. \label{eq:z1a}
\end{eqnarray}
In the following,  we present  an algorithm that produces  a  distribution $T^{*}_{n}=(T^{*}_{n}(a_{1}), \cdots, T^{*}_{n}(a_{k}) )^{T}$ satisfying the following conditions:
\begin{enumerate}
	\item $T^{*}_{n}$ is a type distribution, i.e., for each $ i=1, \cdots,k$,  $n T^{*}_{n}(a_{i}) $ is a non-negative integer and
	\begin{eqnarray}
		\sum_{i=1}^{k}n T^{*}_{n}(a_{i}) =n. \label{eq:z8}  
	\end{eqnarray}
	\item $T^{*}_{n}$  satisfies
	\begin{eqnarray}
		d_{\chi^{2}}(T^{*}_{n}, P) \leq \bar{r}_{n}. \label{eq:z9}
	\end{eqnarray}
	\item $T^{*}_{n}$ further satisfies 
	\begin{eqnarray}
		|n \ell(\Gamma^{*}) -n\ell(T^{*}_{n})| \leq \kappa, \quad \text{ for some } \kappa >0. \label{eq:zz}
	\end{eqnarray}
\end{enumerate} 

Before proceeding to the proof, let us observe that 
\begin{eqnarray}
	d_{\chi^{2}} (\Gamma^{*}, P) &=& \sum_{i=1}^{k}  \frac{(\Gamma^{*}_{i} -P_{i})^{2}}{P_{i}} \notag  \\
	&=&\sum_{i=1}^{k} \frac{ \left(  \sqrt{\bar{r}}\left( D(P \| Q)-\alpha_{i}\right) P_{i}\right)^{2}  }{ V(P\|Q) } \frac{1}{P_{i}} \notag  \\  
	&=&	\frac{\bar{r}}{V(P\|Q)} \sum_{i=1}^{k} \left( D(P \| Q)-\alpha_{i}\right)^{2} P_{i} \notag  \\
	&=&\bar{r}_{n}. \label{eq:dchir}
\end{eqnarray}
Next, it follows from Remark~\ref{remmin} that  there exist two indices $l, l' \in \{1,\cdots,k\}, l\neq l'$ such that $D(P\|Q)-\alpha_{l} < 0$ and $D(P\|Q)-\alpha_{l'} > 0$, which implies that $\Gamma^{*}_{l}<P_{l}$ and $\Gamma^{*}_{l'}>P_{l'} $. Without loss of generality,  let us assume that $l=k-1 $ and $l'=k$, i.e.,
\begin{align}
	(D(P\|Q)-\alpha_{i})P_{i} & <0, \quad i=k-1 \label{eq:Dpq:lwr}\\
	(D(P\|Q)-\alpha_{i})P_{i}  & >0, \quad i=k \label{eq:Dpq:uppr}
\end{align}
and  also that, for some  $0\leq m \leq k-2$, 
\begin{align}
	(D(P\|Q)-\alpha_{i})P_{i} & =0, \quad i=1,\cdots,m,  \label{eq:aldefm0} \\
	(D(P\|Q)-\alpha_{i})P_{i} & \neq 0, \quad i=m+1, \cdots, k-2 \label{eq:aldefmno0} .
\end{align}
Further note that, by \eqref {eq:sec6} and \eqref{eq:rbar},  $\bar{r}_{n}= \Theta(1/n)$. 

By \eqref{eq:z1a}, we have for $i=1,\cdots,k$
\begin{eqnarray}
	n\Gamma^{*}_{i}-  nP_{i} &=& \frac{ n\sqrt{\bar{r}_{n}} (D(P\|Q)-\alpha_{i})P_{i}}{\sqrt{V(P\|Q)} } \label{eq:z1aa} 
\end{eqnarray}
which implies that 
\begin{eqnarray}
	n\Gamma^{*}_{i}- nP_{i} &=& \Theta(\sqrt{n}), \quad i=m+1,\cdots, k. \label{eq:z1aaa}
\end{eqnarray}
Next let 
\begin{align}
	C' & =\sum_{i=1}^{m} \frac{1}{P_{i}}.
\end{align}
It then follows from \eqref{eq:z1aaa}  that, for each $i=m+1,\cdots,k$, there exists an $N_{i}$ such that $|n\Gamma^{*}_{i}- nP_{i}| \geq C'+k$ for  $n \geq N_{i}$. So, for $n \geq \tilde{N} \triangleq \max \{ N_{i} : i=m+1,\cdots,k\}$, we have
\begin{eqnarray}
	|n\Gamma^{*}_{i}- nP_{i}| \geq C'+ k, \quad  i=m+1,\cdots,k.   \label{eq:z0}
\end{eqnarray}
In the following, we describe an algorithm that outputs $T^{*}_{n}$ for all $n \geq \tilde{N}$ satisfying \eqref{eq:z8}, \eqref{eq:z9} and  \eqref{eq:zz}. 
\vspace{2mm}

\noindent{\bf Algorithm} 

\noindent \textit{Step 1}: For $i=1,\cdots, m$, set 
\begin{eqnarray}
	n T^{*}_{n}(a_{i})  &=&
	\lfloor  n \Gamma^{*}_{i} \rfloor\nonumber \\
	&=&\lfloor  n P_{i} \rfloor \label{eq:y1}
\end{eqnarray}
where $\lfloor  \cdot \rfloor $ is the floor function defined as $\lfloor  x \rfloor \triangleq \max\{u \in \mathbb{Z} \; | \; u \leq x  \}$. The second step in \eqref{eq:y1} holds because \eqref{eq:aldefm0} and \eqref{eq:z1aa} imply that $ n \Gamma^{*}_{i}= n P_{i}$ for $i=1,\cdots, m$.

\noindent 	\textit{Step 2}: For $i=m+1,\cdots,k-2$, set
\begin{eqnarray}{\label{eq:y2}}
	n T^{*}_{n}(a_{i})  &=& 
	\begin{cases}
		\lceil n \Gamma^{*}_{i} \rceil, \quad	& \text{if} \; n \Gamma^{*}_{i} -nP_{i} <0 \\
		\lfloor  n \Gamma^{*}_{i} \rfloor, \quad 	&   \text{if}  \;  n \Gamma^{*}_{i} -nP_{i} >0
	\end{cases} \label{eq:lastendi}
\end{eqnarray}
where $\lceil \cdot \rceil$ is the ceiling function defined as  $\lceil   x \rceil  \triangleq \min\{u \in \mathbb{Z} \;| \; u \geq x  \}$.

\noindent	\textit{Step 3}: Let 
\begin{align}
	\gamma_{i}  & \triangleq n \Gamma^{*}_{i} - n T^{*}_{n}(a_{i}), \quad i=1, \cdots, k \label{eq:gammai}\\
	\mathbf{\gamma} & \triangleq \sum_{i=1}^{k-2} \gamma_{i}. \label{eq:gama}
\end{align}
We set  $T^{*}_{n}(a_{k-1})$ or $T^{*}_{n}(a_{k-1})$ depending on whether $\gamma \leq 0$ or $\gamma >0$: \\
\vspace{2mm}
\textit{Case 1}: If $\gamma \leq 0$, set
\begin{eqnarray}
	n T^{*}_{n}(a_{k-1})  &=& 
	\lceil n \Gamma^{*}_{k-1} + C' P_{k-1} \rceil  \label{eq:z4}\\
	n T^{*}_{n}(a_{k}) & =&
	n \Gamma^{*}_{k} + \gamma', \quad \text{where }  \gamma'=\gamma+ \gamma_{k-1} <0. \label{eq:z5}
\end{eqnarray}
\textit{Case 2}: If $\gamma >0$, set
\begin{eqnarray}
	n T^{*}_{n}(a_{k})  &=& 
	\lfloor n \Gamma^{*}_{k} -C' P_{k}\rfloor   \label{eq:z6}\\
	n T^{*}_{n}(a_{k-1}) & =&
	n \Gamma^{*}_{k-1} + \gamma', \quad \text{where }  \gamma'=\gamma+ \gamma_{k} >0. \label{eq:z7}
\end{eqnarray}	

We next demonstrate that $T^{*}_{n}$ produced by this algorithm satisfies \eqref{eq:z8}--\eqref{eq:zz}.

\subsection*{Proof of \eqref{eq:z8}:}
 When $i=1,\cdots,k-2$, since $n T^{*}_{n}(a_{i})$  is defined using the floor or ceiling function, $n T^{*}_{n}(a_{i})$ is clearly a non-negative integer.  For case 1, $n T^{*}_{n}(a_{k-1})$ is also a non-negative  integer since both $n \Gamma^{*}_{k-1}$ and $C'P_{k-1}$ are non-negative numbers and  $n T^{*}_{n}(a_{k-1})$ is defined using  the ceiling function of the sum of these two numbers.  It follows from \eqref{eq:z4} that
\begin{eqnarray}
	n \Gamma^{*}_{k-1} +C'P_{k-1} \leq 	
	n T^{*}_{n}(a_{k-1}) \leq n \Gamma^{*}_{k-1} +C'P_{k-1} +1.  \label{eq:z10}
\end{eqnarray}
Furthermore, by \eqref{eq:z5}, 
\begin{eqnarray}
	|\gamma'| &\leq& \sum_{i=1}^{k-2} |\gamma_{i}| + |\gamma_{k-1}|  \notag \\
	&\leq & k-2 + C'P_{k-1}+1 \label{eq:y0}  \notag \\
	&=& k-1 + C'P_{k-1} \label{eq:z11}
\end{eqnarray}
where the second inequality follows from the definition of $\gamma_i,\; i=1,\cdots, k-1$ and \eqref{eq:z10}. Thus, we have 
\begin{eqnarray}
	-( k-1 )- C'P_{k-1}	\leq \gamma' \leq  (k-1) + C'P_{k-1} \label{eq:z11a}.
\end{eqnarray}
Together with  \eqref{eq:z5}, this yields
\begin{eqnarray}
	n T^{*}_{n}(a_{k}) &=& 
	n \Gamma^{*}_{k} + \gamma' \notag \\
	&	\geq & 	n \Gamma^{*}_{k} -(k-1) -C'P_{k-1} \notag \\
	&> & n \Gamma^{*}_{k} -C'-k \notag \\
	&\geq  & nP_{k} \label{eq:Pklowr}\\
	&\geq& 0 \label{eq:y7}
\end{eqnarray}
where \eqref{eq:Pklowr} follows from \eqref{eq:z0}, \eqref{eq:z1a} and  \eqref{eq:Dpq:uppr}.  We further have that
\begin{align}
	n T^{*}_{n}(a_{k})	=& 	n \Gamma^{*}_{k}+ \gamma' \notag \\
	=& n \Gamma^{*}_{k} + \sum_{i=1}^{k-1} \left( n \Gamma^{*}_{i} -  	n T^{*}_{n}(a_{i}) \right) \notag \\ 
	=& n-\sum_{i=1}^{k-1} n T^{*}_{n}(a_{i}). \label{eq:y8} 
\end{align}
From \eqref{eq:y7}, we obtain that $T^{*}_{n}(a_{k})$ is non-negative. 	Since we have already shown that $n T^{*}_{n}(a_{i}), i=1,\cdots,k-1$ are non-negative integers, it follows from \eqref{eq:y8} that $n T^{*}_{n}(a_{k})$ is an integer, too.  From \eqref{eq:y8}, we also obtain that $\sum_{i=1}^{k} n T^{*}_{n}(a_{i}) =n$. This completes the proof of \eqref{eq:z8} for case 1.

We next prove \eqref{eq:z8} for case 2. From \eqref{eq:z0} and  \eqref{eq:Dpq:uppr},  we have that $n \Gamma^{*}_{k} - nP_k \geq C'+k$. It follows that $nT^{*}_{n}(a_{k}) = \lfloor n \Gamma^{*}_{k}- C'P_k \rfloor $ is non-negative. The fact that $nT^{*}_{n}(a_{k})$ is an integer follows from its definition. Furthermore, it is easy to see that $nT^{*}_{n}(a_{k-1})$ is non-negative by \eqref{eq:z7}. Finally, we have
\begin{eqnarray}
	n T^{*}_{n}(a_{k-1}) &=& 
	n \Gamma^{*}_{k-1}+ \gamma' \notag  \\
	&=& n \Gamma^{*}_{k-1}  + \sum_{i=1}^{k-2} \left( n \Gamma^{*}_{i} -  	n T^{*}_{n}(a_{i}) \right) +  n \Gamma^{*}_{k} - 	n T^{*}_{n}(a_{k}) \notag \\
	&=& n -  \sum_{i=1}^{k-2}	n T^{*}_{n}(a_{i}) - n T^{*}_{n}(a_{k}).  \label{eq:y9}
\end{eqnarray}
Since $n T^{*}_{n}(a_{i}), i=1,\cdots,k-2$ and $n T^{*}_{n}(a_{k})$ are  integers, it follows from \eqref{eq:y9} that $n T^{*}_{n}(a_{k})$ is an integer, too. From \eqref{eq:y9}, we also obtain that $\sum_{i=1}^{k} n T^{*}_{n}(a_{i}) =n$. This demonstrates that $\Gamma^{*}$ satisfies \eqref{eq:z8}.

\subsection*{Proof of \eqref{eq:z9}:} Since, by \eqref{eq:dchir}, $\bar{r}_{n}= d_{\chi^{2}}(\Gamma^{*}, P)$, \eqref{eq:z9} is equivalent to
\begin{eqnarray}
	\sum_{i=1}^{k}\frac{	( nT^{*}_{n}(a_{i})-nP_{i})^{2}}{n^{2}P_{i}}  &\leq& \sum_{i=1}^{k}\frac{	( n\Gamma^{*}_{i}-nP_{i})^{2}}{n^{2}P_{i}}.  \label{eq:dist}
\end{eqnarray}
For $i=1,\cdots,k$, let 
\begin{eqnarray}
	l_{i}'&=& \frac{	( nT^{*}_{n}(a_{i})-nP_{i})^{2}}{n^{2}P_{i}} \\
	l_{i} &=& \frac{	( n\Gamma^{*}_{i}-nP_{i})^{2}}{n^{2}P_{i}}. \label{eq:lidef} 
\end{eqnarray}
Then, \eqref{eq:dist} is equivalent to
\begin{eqnarray}
	\sum_{i=1}^{k}l_{i}' \leq \sum_{i=1}^{k}l_{i}. \label{eq:dist2}
\end{eqnarray}
We have
\begin{eqnarray}
	\sum_{i=1}^{k}l_{i}'&=&  \sum_{i=1}^{m}l_{i}'+\sum_{i=m+1}^{k-2}l_{i}'+l_{k-1}'+l_{k}'. \label{eq:s1}
\end{eqnarray}
By  \eqref{eq:y1}, the first sum can be upper-bounded as
\begin{eqnarray}
	\sum_{i=1}^{m}l_{i}'&=& \sum_{i=1}^{m} \frac{( \lfloor nP_{i}\rfloor-nP_{i})^{2}}{n^{2}P_{i}}  \notag \\
	& \leq & \sum_{i=1}^{m} \frac{1}{n^{2}P_{i}} =\frac{C'}{n^{2}}. \label{eq:u1}
\end{eqnarray} 
 From \eqref{eq:y2}, it follows that, for $i=m+1,\cdots,k-2$, 
\begin{eqnarray}
	(nT^{*}_{n}(a_{i})-nP_{i})^{2} &\leq & (n\Gamma^{*}_{i}-nP_{i})^{2}
\end{eqnarray} 
so the second sum can be upper-bounded as
\begin{eqnarray}
	\sum_{i=m+1}^{k-2}l_{i}' \leq \sum_{i=m+1}^{k-2}l_{i}. \label{eq:u2}
\end{eqnarray}

We next prove \eqref{eq:dist2} for case 1 in the algorithm. From \eqref{eq:z5}, we have
\begin{eqnarray}
	n T^{*}_{n}(a_{k}) -nP_{k} & =&
	n \Gamma^{*}_{k} -nP_{k} + \gamma'.
\end{eqnarray}
Since $\gamma' <0$, we have $n T^{*}_{n}(a_{k}) -nP_{k} < 
n \Gamma^{*}_{k} -nP_{k}$. This, together with the fact that $n T^{*}_{n}(a_{k}) -nP_{k}>0$ (cf. \eqref{eq:Pklowr}),  implies that
\begin{eqnarray}
	l_{k}'  < l_{k}. \label{eq:u3}
\end{eqnarray}
As for index $k-1$, we shall show that
\begin{eqnarray}
	l_{k-1}' \leq l_{k-1}-\frac{C'}{n^{2}} \label{eq:x4}
\end{eqnarray} 
which is equivalent to
\begin{eqnarray}
	(n \Gamma^{*}_{k-1} -nP_{k-1})^{2} -(n T^{*}_{n}(a_{k-1}) -nP_{k-1})^{2} \geq C'P_{k-1}. \label{eq:x3}
\end{eqnarray} 
Let $b=nP_{k-1}- n \Gamma^{*}_{k-1}>0$ and $s=nP_{k-1}- n T^{*}_{n}(a_{k-1})$. Then, we have  
\begin{align}
	s & = n P_{k-1} - \lceil n \Gamma^{*}_{k-1} + C'P_{k-1} \rceil \notag \\
	& > n P_{k-1} -  n \Gamma^{*}_{k-1}  - C'P_{k-1} - 1 \notag \\
	& \geq C'+k -  C'P_{k-1} - 1 \notag \\
	& >0 \label{eq:S4} 
\end{align}
where the first step follows from \eqref{eq:z4},  the third step is  due to \eqref{eq:z0} and the fact that  $n P_{k-1} > n \Gamma^{*}_{k-1}$ (by using \eqref{eq:Dpq:lwr} in \eqref{eq:z1a}), and the last step  follows since $P_{k-1} \leq 1$ and $k\geq 2$. Next note that
\begin{eqnarray}
	(n \Gamma^{*}_{k-1} -nP_{k-1})^{2} -(n T^{*}_{n}(a_{k-1}) -nP_{k-1})^{2} =b^{2}-s^{2}=(b-s)(b+s). 
\end{eqnarray} 
From \eqref{eq:z0}, it follows that $	b \geq  C'+k$, which yields
\begin{eqnarray}
s+b \geq  1 \label{eq:x1}
\end{eqnarray}
since $s>0$ by \eqref{eq:S4} and  $C'+k\geq  1$.Furthermore, \eqref{eq:z10} yields that
\begin{eqnarray}
	b-s &=& n T^{*}_{n}(a_{k-1}) -n \Gamma^{*}_{k-1} \geq C' P_{k-1}.\label{eq:x2}
\end{eqnarray}
We thus obtain \eqref{eq:x3} from \eqref{eq:x1} and \eqref{eq:x2}. It follows from \eqref{eq:u1}, \eqref{eq:u2}, \eqref{eq:u3}, and  \eqref{eq:x4} that, for case 1 in the algorithm,
\begin{eqnarray}
	\sum_{i=1}^{k}l_{i}'&=&  \sum_{i=m+1}^{k-2}l_{i}'+\sum_{i=1}^{m}l_{i}'+ l_{k-1}'+l_{k}' \notag \\
	&\leq & \sum_{i=m+1}^{k-2}l_{i} + \frac{C'}{n^{2}}+  l_{k-1}-\frac{C'}{n^{2}} +l_{k} \notag  \\
	&= & \sum_{i=m+1}^{k}l_{i}  \notag \\
	&=& \sum_{i=1}^{k}l_{i} \label{eq:a4} 
\end{eqnarray}
where the last step follows since, by  \eqref{eq:z1a},  \eqref{eq:aldefm0}, and \eqref{eq:lidef}, we have $\sum_{i=1}^{m}l_{i} =0$.

We next  argue that, by following similar steps as above, we also obtain \eqref{eq:dist2} for case 2. To this end, we use that
\begin{align}
	n \Gamma^{*}_{k-1} <	n T^{*}_{n}(a_{k-1})  < nP_{k-1} \label{eq:Tn:Pn}
\end{align}
where the first inequality follows  since	$ n T^{*}_{n}(a_{k-1}) = n \Gamma^{*}_{k-1}+ \gamma'$ and $\gamma' >0$ (by \eqref{eq:z7}), and the second inequality follows since 
\begin{align}
	nP_{k-1} & \geq n\Gamma^{*}_{k-1}+ C'+k \notag \\
	& >  n\Gamma^{*}_{k-1} + \gamma' \notag \\
	& =  n T^{*}_{n}(a_{k-1}). \label{eq:Pk3}
\end{align}
Here, the first step is due to  \eqref{eq:z0} and the fact that $nP_{k-1} \geq n\Gamma^{*}_{k-1}$ (by using \eqref{eq:Dpq:lwr} in \eqref{eq:z1a}), the second step follows   from the fact that $\gamma' < C'+k$ (by \eqref{eq:z11}), and the last step follows from  \eqref{eq:z7}. By \eqref{eq:Tn:Pn}, we have  that $|n T^{*}_{n}(a_{k-1}) -nP_{k-1} | < |n \Gamma^{*}_{k-1} -nP_{k-1} | $, which implies that
\begin{eqnarray}
	l_{k-1}'  < l_{k-1}. \label{eq:v3}
\end{eqnarray}
We finally obtain 
\begin{eqnarray}
	l_{k}' < l_{k}-\frac{C'}{n^{2}} \label{eq:v4}
\end{eqnarray} 
by letting $b=(n \Gamma^{*}_{k} -nP_{k})$, $s=(n T^{*}_{n}(a_{k}) -nP_{k})$, and  by then following the steps that led to \eqref{eq:x2}.
Thus, \eqref{eq:dist2} follows from \eqref{eq:u1}, \eqref{eq:u2}, \eqref{eq:v3}, and  \eqref{eq:v4}, and by using that $\sum_{i=1}^m l_i =0$. This demonstrate that $\Gamma^{*}$ satisfies \eqref{eq:z9}.

\subsection*{ Proof of \eqref{eq:zz}:} 
For case 1 in the algorithm, we obtain from \eqref{eq:lastendi}, \eqref{eq:z10}, \eqref{eq:z11}, and the definition of $\ell(\cdot)$ in \eqref{eq:hfun1}
\begin{eqnarray}
	|n \ell(\Gamma^{*}) -n\ell(T^{*}_{n})| & = & \left|   \sum_{i=1}^{k} (n\Gamma^{*}_{i}-nT^{*}_{n}(a_{i})) \alpha_{i}\right| \notag  \\
	& \leq & \sum_{i=1}^{k} |  n\Gamma^{*}_{i}-nT^{*}_{n}(a_{i}) |  |\alpha_{i}|  \notag \\
	& = & \sum_{i=1}^{k-2} |  n\Gamma^{*}_{i}-nT^{*}_{n}(a_{i}) |  | \alpha_{i}| + |n\Gamma^{*}_{k-1}-nT^{*}_{n}(a_{k-1}) |  |\alpha_{k-1}|  \notag \\
	& & \hspace{20mm} + |n\Gamma^{*}_{k}-nT^{*}_{n}(a_{k}) |  |\alpha_{k}| \notag \\
	& \leq&   \sum_{i=1}^{k-2} | \alpha_{i}| + |\alpha_{k-1}| (C'P_{k-1}+1) + |\alpha_{k}| (C'P_{k-1}+k-1) \notag \\
	&\triangleq & \kappa_{1}
\end{eqnarray}
where  $\kappa_{1} $ is a positive constant.

Similarly, for case 2 in the algorithm, we have that
\begin{eqnarray*}
	|n \ell(\Gamma^{*}) -n\ell(T^{*}_{n})| & = & \left|   \sum_{i=1}^{k} (n\Gamma^{*}_{i}-nT^{*}_{n}(a_{i})) \alpha_{i}\right|  \\
	& \leq & \sum_{i=1}^{k} |  n\Gamma^{*}_{i}-nT^{*}_{n}(a_{i}) |  |\alpha_{i}| \\
	& = & \sum_{i=1}^{k-2} |  n\Gamma^{*}_{i}-nT^{*}_{n}(a_{i}) |  | \alpha_{i}| + |n\Gamma^{*}_{k-1}-nT^{*}_{n}(a_{k-1}) |  |\alpha_{k-1}|  \\
	& & \hspace{20mm} + |n\Gamma^{*}_{k}-nT^{*}_{n}(a_{k}) |  |\alpha_{k}|\\
	& \leq&   \sum_{i=1}^{k-2} | \alpha_{i}| + |\alpha_{k-1}|(C'P_{k}+k-1) + |\alpha_{k}|(C'P_{k}+1)  \\
	&\triangleq & \kappa_{2} 
\end{eqnarray*} 
where  $\kappa_{2} $ is a positive constant.

By taking $\kappa = \max\{\kappa_{1}, \kappa_{2}   \}$, we obtain
\begin{eqnarray*}
	|n \ell(\Gamma^{*}) -n\ell(T^{*}_{n})| \leq \kappa
\end{eqnarray*}  
thus demonstrating that $\Gamma^{*}$ satisfies \eqref{eq:zz}.

 \section*{Acknowledgment}
 \addcontentsline{toc}{section}{Acknowledgment}
 The authors would like to thank Shun Watanabe for his helpful comments during the initial phase of this work and for pointing out relevant literature.

\bibliography{Bibliography.bib}
\bibliographystyle{IEEEtran}

\end{document}